\documentclass[12pt, draftclsnofoot, onecolumn]{IEEEtran}
\usepackage[dvips]{graphicx}
\usepackage{amsmath}
\usepackage{amssymb}
\usepackage{citesort}
\usepackage{graphicx}
\usepackage{times}
\usepackage{amsfonts}
\usepackage{float}
\usepackage{amsthm}
\usepackage{color}
\usepackage{psfrag}
\usepackage{setspace}
\newtheorem{thrm}{Theorem}

\newtheorem{corol}{Corollary}
\newtheorem{pro}{Proposition}
\newtheorem{remark}{Remark}

\begin{document}

\title{Jamming-Resistant Receivers for the Massive MIMO Uplink}

\author{Tan Tai Do$^\dagger$, Emil Bj\"ornson$^\dagger$, Erik G. Larsson$^\dagger$, and S. Mohammad Razavizadeh$^*$\\
{\small $^\dagger$ Department of Electrical Engineering, Link\"oping University (LiU), Sweden} \\
{\small $^*$ School of Electrical Engineering, Iran University of Science $\&$ Technology (IUST), Iran}

\thanks{This work will be partly presented at the 42nd IEEE Int. Conf. on Acoustics, Speech and Signal Process.
(ICASSP2017), Mar. 5-9, 2017, New Orleans, USA \cite{DBL17ICASSP}.}
}

\maketitle

\begin{abstract}
We design a jamming-resistant receiver scheme to enhance the robustness of a massive MIMO uplink system against jamming. We assume that a jammer attacks the system both in the pilot and data transmission phases.
The key feature of the proposed scheme is that, in the pilot phase, we estimate not only the legitimate channel, but also the jamming channel by exploiting a purposely unused pilot sequence. The jamming channel estimate is used to constructed  linear receive filters that reject the impact of the jamming signal. The performance of the proposed scheme is analytically evaluated using asymptotic properties of massive MIMO. The optimal regularized zero-forcing receiver and the optimal power allocation are also studied. Numerical results are provided to verify our analysis and show that the proposed scheme greatly improves the achievable rates, as compared to conventional receivers. Interestingly, the proposed scheme works particularly well under strong jamming attacks, since the improved estimate of the jamming channel outweighs the extra jamming power.
\end{abstract}

\begin{IEEEkeywords}
Massive MIMO, jamming attack, receive filter, optimal power allocation.
\end{IEEEkeywords}
%%%%%%%%%%%%%%%%%%%%%%%%%%%%%%%%%%%%%%%%%%%%%%%%%%%%%%%%%%%%%%%%%%%%%%%%%%%%%%%%%%%%%%%%%%%

%%%%%%%%%%%%%%%%%%%%%%%%%%%%%%%%%%%%%%%%%%%%%%%%%%%%%%%%%%%%%%%%%%%%%%%%%%%%%%%%%%%%%%%%%%%
\section{Introduction}\label{sec:intro}
%%%%%%%%%%%%%%%%%%%%%%%%%%%%%%%%%%%%%%%%%%%%%%%%%%%%%%%%%%%%%%%%%%%%%%%%%%%%%%%%%%%%%%%%%%%
As a promising candidate for the emerging 5G wireless communication networks \cite{MLYNb}, massive multiple-input multiple-output (MIMO) has recently received a lot of research attention. However, the physical layer security of massive MIMO is not well studied in literature. A potential reason is the belief that physical layer security techniques in regular MIMO systems, which have been intensively studied, can be straight-forwardly extended to massive MIMO systems. However, as shown in \cite{RZR15CM,BKA15CNS}, the spatial dimensions that massive MIMO exploits bring new challenges and opportunities to the physical layer security, which are fundamentally different from conventional MIMO systems. While massive MIMO is robust against passive eavesdropping \cite{RZR15CM} thanks to the array gain and ability to operate with lower transmit power, active jamming attacks is a big challenge. When a massive MIMO system is attacked by jamming, especially in the pilot phase, the additional pilot contamination caused by the jamming leads to an inability to suppress the jamming and this translates into a significant reduction of the achievable rates \cite{BKA15CNS}.

Although jamming exists and has been identified as a critical problem for reliable communications \cite{WoS02Comp,YZJ16TIFS}, there are only a few works focusing on the jamming aspects in massive MIMO \cite{RZR15CM,BKA15CNS,IJC13GBC,WLW15TWC,PRB16WCL,WSN16TIT}. For instance, optimized jamming is studied for uplink massive MIMO in \cite{PRB16WCL}, which shows that a smart jammer can cause substantial jamming pilot contamination that degrades the sum rate. In order to detect pilot contamination attacks in massive MIMO, several jamming detection techniques are introduced in \cite{RZR15CM,IJC13GBC,PRB16WCL}. Jamming defense mechanisms for massive MIMO are proposed in \cite{BKA15CNS,IJC13GBC}, in which secret keys are employed to encrypt and protect the legitimate signal from jamming attacks. The authors of \cite{WSN16TIT} investigates an artificial noise-aided transmitter for secure communications in the presence of attackers capable of both jamming and eavesdropping.

Pilot contamination appears when the pilot signal, transmitted for estimation of a user channel, is interfered by another transmission \cite{Mar10TWC}. The typical effect is that the base station (BS) cannot use the estimated channel to coherently combine the desired signal, without also coherently combining the interference. Pilot contamination between legitimate users of the system is a big challenge in massive MIMO, but can be substantially suppressed by pilot coordination across cells \cite{Huh2012a,BLD16TWC} or by exploiting second-order channel statistics \cite{Yin2013a,Bjornson2017a}. Jamming pilot contamination is more difficult to deal with, because the jammer refuses to coordinate itself with the system and attempts to create maximum pilot contamination rather than minimum. Since the knowledge of the structure and properties of the jamming attack is limited, a typical approach to deal with jamming signals is to treat them as additive noise and design the transceivers as if there was no jamming \cite{BKA15CNS,PRB16WCL}. However, jamming in massive MIMO is not noise-like since the legitimate channel estimate is correlated with the jamming channel.

In order to enhance the robustness of the massive MIMO uplink against jamming attacks, we consider an anti-jamming scheme based on jamming-resistant receivers, which is briefly introduced in \cite{DBL17ICASSP}. Developed from our initial concept in \cite{DBL17ICASSP}, this paper provides a comprehensive study of jamming-resistant receiver design for the massive MIMO uplink by including various rigorous proofs and new results related to the effect of jamming powers, optimal receive filters, and power allocations.

The key idea of the proposed scheme is to construct the receive filters using not only an estimate of the legitimate channel but also an estimate of the jamming channel. To this end, we exploit a purposely unused pilot sequence, which is orthogonal to the pilot sequences assigned to the legitimate user, to estimate the jamming channel, up to an unknown scaling factor. This estimate is used to design receive filters that reject the jamming signal. We consider regularized zero-forcing (RZF) receive filters, in which the regularization factor can be adjusted to optimize the system performance.

To evaluate the performance of the proposed scheme, the achievable rates are analyzed and closed-form large-scale approximations are obtained. Based on the analytical results, we derive the optimal regularization factor for the RZF receiver. Moreover, we study how a legitimate user should allocate its power between the pilot and data phases. We obtain an asymptotically optimal power allocation for systems with a very large number of antennas and a sub-optimal power allocation for cases with a finite number of antennas. Simulation results are provided, which reveal that the proposed jamming-resistant receivers and power allocation  substantially improve the system performance over conventional schemes.

The rest of paper is organized as follows. Section~\ref{sec:probsetup} presents the problem setup and signal models for the pilot and data transmission phases. Section~\ref{sec:ch_est} considers the channel estimations and jamming-resistant receiver design. In Section~\ref{sec:opt_rx}, closed-form large-scale approximations for the achievable rates of the proposed scheme are provided. The optimal RZF receiver and effects of the jamming powers are also analyzed in this section. Section~\ref{sec:opt_p} studies the optimal power allocation. Numerical results are then provided in Section~\ref{sec:numerical} and the main conclusions are given in Section~\ref{sec:conclusion}.

%%%%%%%%%%%%%%%%%%%%%%%%%%%%%%%%%%%%%%%%%%%%%%%%%%%%%%%%%%%%%%%%%%%%%%%%%%%%%%%%%%%%%%%%%%%

%%%%%%%%%%%%%%%%%%%%%%%%%%%%%%%%%%%%%%%%%%%%%%%%%%%%%%%%%%%%%%%%%%%%%%%%%%%%%%%%%%%%%%%%%%%
\section{Problem Setup}\label{sec:probsetup}
%%%%%%%%%%%%%%%%%%%%%%%%%%%%%%%%%%%%%%%%%%%%%%%%%%%%%%%%%%%%%%%%%%%%%%%%%%%%%%%%%%%%%%%%%%%

We consider a single-user massive MIMO uplink consisting of a BS, a legitimate user and a jammer, as depicted in Fig.~\ref{fig:system_model}. We assume that the BS is equipped with $M$ antennas, while the legitimate user and the jammer have a single antenna each. This basic model captures the main principle of jamming, and the methodology can be generalized to having multiple legitimate users.

\begin{figure}[t]
\centering
\psfrag{BS}[][][0.7]{$\mathrm{Base~station}$}
\psfrag{US}[][][0.7]{$\mathrm{User}$}
\psfrag{JM}[][][0.7]{$\mathrm{Jammer}$}
\psfrag{gu}[][][0.7]{$\mathbf{h}$}
\psfrag{gj}[][][0.7]{$\mathbf{g}$}
\includegraphics[width=7.5cm]{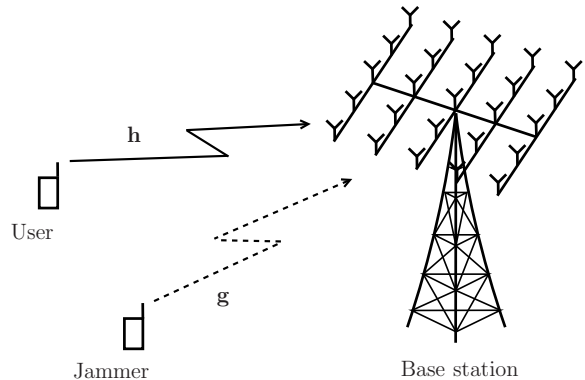}
\caption{Massive MIMO uplink under a jamming attack.}\label{fig:system_model}
\end{figure}

Let us denote $\mathbf{h}\in \mathbb{C}^{M\times 1}$ and $\mathbf{g}\in \mathbb{C}^{M\times 1}$ as the channel vectors from the legitimate user and the jammer to the BS, respectively. We assume that the elements of $\mathbf{h}$ are independent and identically distributed (i.i.d.) zero-mean circularly-symmetric complex Gaussian (ZMCSCG) random variables, i.e., $\mathbf{h}\sim \mathcal{CN}(0,\beta_\mathrm{u}\mathbf{I}_M)$, where the variance $\beta_\mathrm{u}$ represents the large-scale fading. Similarly, we assume that $\mathbf{g}\sim \mathcal{CN}(0,\beta_\mathrm{j}\mathbf{I}_M)$, where the variance $\beta_\mathrm{j}$ represents the large-scale fading. The channels $\mathbf{h}$ and $\mathbf{g}$ are independent.

We consider a block-fading model, in which the channel remains constant during a coherence block of $T$ symbols, and varies independently from one coherence block to the next. The communication between the legitimate user and the BS follows a two-phase transmission protocol. In the first phase (pilot phase), the legitimate user sends a pilot sequence to the BS for channel estimation. In the second phase (data transmission phase), the legitimate user transmits its payload data to the BS. We assume that the jammer attacks the system both in the pilot and data transmission phases.

%%%%%%%%%%%%%

\subsection{Pilot Phase}
During the first $\tau$ symbols of a coherence block ($\tau<T$), the user transmits a pilot sequence $\mathbf{s}_\mathrm{u}$ of length $\tau$ symbols. This pilot originates from a pilot codebook $\mathcal{S}$ containing $\tau$ orthogonal unit-power vectors. We assume that there is (at least) one pilot sequence that is unused and orthogonal to pilot sequences assigned to the legitimate user.\footnote{This is often the case in real systems, when the system is dimensioned for a maximum simultaneous user load that is substantially larger than the number of active users at most points in time.} We further assume that the jammer is aware of the transmission protocol but the legitimate system uses a pilot hopping scheme such that the jammer cannot know the users' current pilot sequences. Therefore, the jammer randomly chooses a jamming sequence $\mathbf{s}_{\mathrm{j}}$ uniformly distributed over the unit sphere. By sending the jamming sequence $\mathbf{s}_{\mathrm{j}}\in \mathbb{C}^{\tau \times 1}$, which satisfies $\|\mathbf{s}_{\mathrm{j}}\|^2=1$, the jammer hopes to interfere with the channel estimation.

Accordingly, the received signal at the $M$ antennas of the BS in the $\tau$ symbol times of the pilot phase can be stacked as
\begin{align}
\label{ytma}
\mathbf{Y}_\mathrm{t}=\sqrt{\tau p_\mathrm{u}}\mathbf{h}\mathbf{s}_\mathrm{u}^T+\sqrt{\tau p_\mathrm{j}}\mathbf{g}\mathbf{s}_\mathrm{j}^T+\mathbf{N}_\mathrm{t},
\end{align}
where $\mathbf{Y}_\mathrm{t}\in \mathbb{C}^{M \times \tau}$, $p_\mathrm{u}$ and $p_\mathrm{j}$ are the transmit powers per symbol of the user and jammer during the pilot phase, respectively. The additive noise matrix $\mathbf{N}_\mathrm{t} \in \mathbb{C}^{M \times \tau}$ is assumed to have i.i.d.~ZMCSCG elements, i.e., $\mathrm{vec}(\mathbf{N}_\mathrm{t})\sim \mathcal{CN}(0,\sigma^2\mathbf{I}_{M^2\tau^2})$, where $\sigma^2$ is the noise variance and $\mathrm{vec}(\mathbf{N}_\mathrm{t})$ is the vectorization of $\mathbf{N}_\mathrm{t}$.

\subsection{Data Transmission Phase}
During the last $(T-\tau)$ symbols of a coherence block, the user transmits payload data to the BS and the jammer continues to interfere by sending a jamming signal. Let us denote as $x_\mathrm{u}$ and $x_\mathrm{j}$ the transmitted signals
from the user and the jammer, respectively. These signals satisfy $\mathbb{E}\{|x_\mathrm{u}|^2\}=1$ and $\mathbb{E}\{|x_\mathrm{j}|^2\}=1$.
The received signal at the BS is
\begin{align}
\label{yd}
\mathbf{y}_\mathrm{d}=\sqrt{q_\mathrm{u}}\mathbf{h}x_\mathrm{u}+\sqrt{q_\mathrm{j}}\mathbf{g}x_\mathrm{j}+\mathbf{n}_\mathrm{d},
\end{align}
where the additive noise vector $\mathbf{n}_\mathrm{d}$ is assumed to have i.i.d.~$\mathcal{CN}(0,\sigma^2)$ elements, $q_\mathrm{u}$ and $q_\mathrm{j}$ are the transmit powers from the user and jammer in the data transmission phase, respectively.

To detect $x_\mathrm{u}$ based on $\mathbf{y}_\mathrm{d}$, the BS uses a linear receive filter as follows:
\begin{align}
\label{y}
y=\mathbf{a}^H\mathbf{y}_\mathrm{d}=\sqrt{q_\mathrm{u}}\mathbf{a}^H\mathbf{h}x_\mathrm{u}+
\sqrt{q_\mathrm{j}}\mathbf{a}^H\mathbf{g}x_\mathrm{j}+\mathbf{a}^H\mathbf{n}_\mathrm{d},
\end{align}
where $\mathbf{a}\in \mathbb{C}^{M \times 1}$ is the receive filter, which will be carefully selected in the next section to reject the jamming.
The received signal in (\ref{y}) can be rewritten as
\begin{align}
\label{y2}
y=\sqrt{q_\mathrm{u}}\mathbb{E}\{\mathbf{a}^H\mathbf{h} |\mathbf{s}_\mathrm{j}\}x_\mathrm{u}+\sqrt{q_\mathrm{u}}(\mathbf{a}^H\mathbf{h}-
\mathbb{E}\{\mathbf{a}^H\mathbf{h} |\mathbf{s}_\mathrm{j} \})x_\mathrm{u}+\sqrt{q_\mathrm{j}}\mathbf{a}^H\mathbf{g}x_\mathrm{j}+\mathbf{a}^H\mathbf{n}_\mathrm{d}.
\end{align}
By treating $\sqrt{q_\mathrm{u}}\mathbb{E}\{\mathbf{a}^H\mathbf{h} |\mathbf{s}_\mathrm{j} \}$ as the deterministic channel that the desired signal is received over and treating the last three terms (which are uncorrelated with $x_\mathrm{u}$) as worst-case independent Gaussian noise, an achievable rate for the legitimate user in the massive MIMO uplink is
\begin{align}
\label{R}
R= \left(1-\frac{\tau}{T}\right) \mathbb{E}_{\mathbf{s}_\mathrm{j}}\left\{\log_2\left(1+\rho\right)\right\},
\end{align}
where the pre-log factor $\left(1-\frac{\tau}{T}\right)$ accounts for the channel estimation overhead and $\rho$ is the effective signal-to-interference-and-noise ratio (SINR), which is given by
\begin{align}
\label{rho}
\rho=\frac{q_\mathrm{u}|\mathbb{E}\{\mathbf{a}^H\mathbf{h}|\mathbf{s}_\mathrm{j}\}|^2}{q_\mathrm{u}\texttt{var}\{\mathbf{a}^H\mathbf{h}|\mathbf{s}_\mathrm{j}\}
+q_\mathrm{j}\mathbb{E}\{|\mathbf{a}^H\mathbf{g}|^2|\mathbf{s}_\mathrm{j}\}+\sigma^2\mathbb{E}\{\|\mathbf{a}\|^2|\mathbf{s}_\mathrm{j}\}}
\end{align}
and $\texttt{var}\{\mathbf{a}^H\mathbf{h}|\mathbf{s}_\mathrm{j}\}=\mathbb{E}\{|\mathbf{a}^H\mathbf{h}|^2|\mathbf{s}_\mathrm{j}\}-|\mathbb{E}\{\mathbf{a}^H\mathbf{h}|\mathbf{s}_\mathrm{j}\}|^2$. In (\ref{rho}), the numerator ($q_\mathrm{u}|\mathbb{E}\{\mathbf{a}^H\mathbf{h}|\mathbf{s}_\mathrm{j}\}|^2$) represents the effective desired signal power. The first term ($q_\mathrm{u}\texttt{var}\{\mathbf{a}^H\mathbf{h}|\mathbf{s}_\mathrm{j}\}$), second term ($q_\mathrm{j}\mathbb{E}\{|\mathbf{a}^H\mathbf{g}|^2|\mathbf{s}_\mathrm{j}\}$), and third term ($\sigma^2\mathbb{E}\{\|\mathbf{a}\|^2|\mathbf{s}_\mathrm{j}\}$) in the denominator of (\ref{rho}) correspond to the undesired signals' powers resulted from the channel uncertainty, jamming, and additive noise, respectively. Typically, due to the channel hardening in massive MIMO, the power terms associated with the channel uncertainty and additive noise are negligible compared to the desired signal term and the jamming term.\footnote{We will later show that the power terms associated with the channel uncertainty and additive noise are proportional to the number of antennas $M$, whereas the desired signal term and the jamming term are proportional to $M^2$ when $M$ is large.} In order to improve the system performance, one can focus on selecting the receive filter $\mathbf{a}$ such that it amplifies the desired signal ($q_\mathrm{u}|\mathbb{E}\{\mathbf{a}^H\mathbf{h}|\mathbf{s}_\mathrm{j}\}|^2$ is large), while mitigating the jamming signal ($q_\mathrm{j}\mathbb{E}\{|\mathbf{a}^H\mathbf{g}|^2|\mathbf{s}_\mathrm{j}\}$ is as small as possible).

\begin{remark}
We note that the expectations in (\ref{rho}) are with respect to (w.r.t.) $\mathbf{h}$, $\mathbf{g}$, $\mathbf{N}_\mathrm{t}$, $x_\mathrm{u}$, and $x_\mathrm{j}$. The effective SINR $\rho$ in (\ref{rho}) is conditioned on $\mathbf{s}_\mathrm{j}$. In order to realize the achievable rate in (\ref{R}), the BS needs to know the numerator and the denominator of $\rho$. Although $\mathbf{s}_\mathrm{j}$ is assumed to be unknown by the system, we will later show that $\rho$ only depends on the correlations of $\mathbf{s}_\mathrm{j}$ and the legitimate pilot sequences, which can be estimated with high accuracy thanks to the asymptotic properties of massive MIMO. We also stress that the receiver processing proposed in the next section will not exploit any instantaneous knowledge of $\mathbf{s}_\mathrm{j}$.
\end{remark}

%%%%%%%%%%%%%%%%%%%%%%%%%%%%%%%%%%%%%%%%%%%%%%%%%%%%%%%%%%%%%%%%%%%%%%%%%%%%%%%%%%%%%%%%%%%

%%%%%%%%%%%%%%%%%%%%%%%%%%%%%%%%%%%%%%%%%%%%%%%%%%%%%%%%%%%%%%%%%%%%%%%%%%%%%%%%%%%%%%%%%%%
\section{Channel estimation and jamming-resistant receiver design}\label{sec:ch_est}
%%%%%%%%%%%%%%%%%%%%%%%%%%%%%%%%%%%%%%%%%%%%%%%%%%%%%%%%%%%%%%%%%%%%%%%%%%%%%%%%%%%%%%%%%%%

The achievable rate in (\ref{R}) highly depends on the choice of the receive filter $\mathbf{a}$. To harvest an array gain, it should be selected as a function of the received pilot signal in the pilot phase. In this section, we propose a jamming-resistant receive filter, which is constructed based on not only the estimate of the legitimate channel but also on an estimate of the jamming channel.

\subsection{Channel Estimation}
In order to estimate the legitimate channel $\mathbf{h}$, the received pilot signal $\mathbf{Y}_\mathrm{t}$ is first correlated with the user's pilot sequence $\mathbf{s}_\mathrm{u}$ as
\begin{align}
\label{yth}
\mathbf{y}_\mathrm{u}=\mathbf{Y}_\mathrm{t}\mathbf{s}_\mathrm{u}^*=\sqrt{\tau p_\mathrm{u}}\mathbf{h}+\sqrt{\tau p_\mathrm{j}}\mathbf{s}_\mathrm{j}^T\mathbf{s}_\mathrm{u}^*\mathbf{g}+ \mathbf{N}_\mathrm{t}\mathbf{s}_\mathrm{u}^*.
\end{align}
Since the BS does not know $\mathbf{s}_\mathrm{j}^T\mathbf{s}_\mathrm{u}^*$, but only its distribution (as explained below), the linear MMSE estimate of $\mathbf{h}$ given $\mathbf{y}_\mathrm{t}$ is \cite{Kay93}
\begin{align}
\label{he}
\mathbf{\widehat{h}}=\eta_\mathrm{u}\mathbf{y}_\mathrm{t}\triangleq\alpha_1 \mathbf{h}+\alpha_2 \mathbf{g}+ \mathbf{n}_1,
\end{align}
where $\eta_\mathrm{u}=\frac{\sqrt{\tau p_\mathrm{u}}\beta_\mathrm{u}}{\tau p_\mathrm{u} \beta_\mathrm{u}+p_\mathrm{j}\beta_\mathrm{j}+1}$, $\alpha_1=\eta_\mathrm{u}\sqrt{\tau p_\mathrm{u}}$, $\alpha_2=\eta_\mathrm{u}\sqrt{\tau p_\mathrm{j}}\mathbf{s}_\mathrm{j}^T\mathbf{s}_\mathrm{u}^*$, and $\mathbf{n}_1 \sim \mathcal{CN}(0,\eta_\mathrm{u}^2\sigma^2\mathbf{I}_M)$. In order to perform the estimation in (\ref{he}), the BS needs to know $p_\mathrm{j}\beta_\mathrm{j}$. Since this term contains the large-scale fading $\beta_\mathrm{j}$, it changes very slowly with time (e.g., some 40 times slower than the small-scale fading according to \cite{Rap96WCP,AML14TIT}). Therefore, the BS can estimate $p_\mathrm{j}\beta_\mathrm{j}$ in advance, e.g., the BS can let the user be silent at some random occasions (unknown to the jammer) to measure the corresponding power level or use blind estimation techniques \cite{AML14TIT}.

As we can see from (\ref{he}), the legitimate channel estimate $\mathbf{\widehat{h}}$ is correlated with the jamming channel $\mathbf{g}$. Without the knowledge of the jamming channel $\mathbf{g}$, the receive filter $\mathbf{a}$ is generally chosen as a linear function of $\mathbf{\widehat{h}}$. In the absence of jamming, the optimal receive filter is
\begin{equation} \label{MRC-definition}
\mathbf{a}_{\mathrm{MRC}} = \mathbf{\widehat{h}},
\end{equation}
which is known as maximal ratio combining (MRC). If such a receive filter is used heuristically in the presence of jamming, the correlation with the jamming channel (in the sense that $\mathbb{E}\{ \mathbf{g}^H \mathbf{\widehat{h}}|\mathbf{s}_\mathrm{j}\} = M\alpha_2\beta_\mathrm{j}$) results in an amplification of the jamming signal. This leads to a degradation of the system performance \cite{BKA15CNS}, known as jamming pilot contamination. In order to mitigate this effect, we propose to design the receive filter based on both $\mathbf{h}$ and $\mathbf{g}$. However, since $\mathbf{h}$ and $\mathbf{g}$ are not available at the BS, we construct receive filters using their estimates instead.

Recall that there is (at least) one unused pilot sequence, preserved in the system, which is orthogonal to the user's pilot $\mathbf{s}_\mathrm{u}$. By projecting the received pilot signal $\mathbf{Y}_\mathrm{t}$ onto this unused pilot sequence, the user's pilot signal is eliminated, leaving only the jamming signal (and noise). The resulting signal is
\begin{align}
\label{ytg}
\mathbf{y}_\mathrm{j}=\mathbf{Y}_\mathrm{t}\mathbf{s}_\mathrm{\overline{u}}^*=\sqrt{\tau p_\mathrm{j}}\mathbf{s}_\mathrm{j}^T\mathbf{s}_\mathrm{\overline{u}}^*\mathbf{g}+ \mathbf{N}_\mathrm{t}\mathbf{s}_\mathrm{\overline{u}}^*,
\end{align}
where $\mathbf{s}_\mathrm{\overline{u}}$ is the unused pilot sequence, satisfying $\mathbf{s}_\mathrm{u}^T\mathbf{s}_\mathrm{\overline{u}}^*=0$. An estimate of the jamming channel $\mathbf{g}$ can be obtained as
\begin{align}
\label{ge}
\mathbf{\widehat{g}}=\eta_{\mathrm{j}}\mathbf{y}_\mathrm{j}\triangleq \alpha_3 \mathbf{g}+ \mathbf{n}_2,
\end{align}
where $\eta_{\mathrm{j}}=\frac{1}{{\sqrt{p_\mathrm{j}\beta_\mathrm{j}+\sigma^2}}}$, $\alpha_3=\eta_{\mathrm{j}}\sqrt{\tau p_\mathrm{j}}\mathbf{s}_\mathrm{j}^T\mathbf{s}_\mathrm{\overline{u}}^*$, and $\mathbf{n}_2 \sim \mathcal{CN}(0,\eta_{\mathrm{j}}^2\sigma^2\mathbf{I}_M)$. We note that the quality of the jamming channel estimate $\mathbf{\widehat{g}}$ depends on the value of $|\mathbf{s}_\mathrm{j}^T\mathbf{s}_\mathrm{\overline{u}}^*|$. It can happen that $|\mathbf{s}_\mathrm{j}^T\mathbf{s}_\mathrm{\overline{u}}^*|=0$, but the probability is zero since $\mathbf{s}_\mathrm{j}$ is a continuous random vector, which is uniformly distributed over the unit sphere. Moreover, when there are more than one unused pilot sequence, one can improve the quality of the jamming channel estimate $\mathbf{\widehat{g}}$ by selecting the unused pilot sequence, which maximizes $|\mathbf{s}_\mathrm{j}^T\mathbf{s}_\mathrm{\overline{u}}^*|$ or combining all of them.
In this paper, we pick one unused pilot at random, to focus on the basic behaviors, and leave potential improvements for future work.

\subsection{Jamming-Resistant Receiver}

Based on the estimates $\mathbf{\widehat{h}}$ and $\mathbf{\widehat{g}}$, we now construct a jamming-resistant receiver, which is inspired by the RZF receiver \cite{PHS05TCOM}. Accordingly, we propose the following RZF receiver\footnote{Conventionally, the RZF receiver is constructed based on the exact channels $\mathbf{h}$ and $\mathbf{g}$, in which case it becomes $( \mathbf{h}\mathbf{h}^H + \mathbf{g}\mathbf{g}^H+\mu\mathbf{I}_M)^{-1} [\mathbf{h}  \, \,\mathbf{g}]$. Since we are not interested in the jamming signal, we only need the first column: $( \mathbf{h}\mathbf{h}^H + \mathbf{g}\mathbf{g}^H+\mu\mathbf{I}_M)^{-1} \mathbf{h}$ which is proportional to $( \mathbf{g}\mathbf{g}^H+\mu\mathbf{I}_M)^{-1} \mathbf{h}$ due to the matrix inversion lemma \cite{HBD13JSAC}. From this expression, with a slightly abuse of definition, we call the receive filter in (\ref{arzf}) an RZF receiver under imperfect channel knowledge.}
\begin{align}
\label{arzf}
\mathbf{a}=\mathbf{a}_{\mathrm{RZF}}=\left(\mathbf{\widehat{g}}\mathbf{\widehat{g}}^H+\mu\mathbf{I}_M\right)^{-1}\mathbf{\widehat{h}},
\end{align}
where $\mu\geq0$ is the regularization factor, determining the amount of interference power remaining at the receiver \cite{PHS05TCOM}. By adjusting $\mu$, one can balance between the targets of amplifying the desired signal and mitigating the undesired jamming signal to optimize the overall system performance. In the following, we consider two common examples of the RZF receiver.

\subsubsection{MMSE-type receiver}\label{subsec:mmse}
First, we consider the MMSE receive filter, which is optimal when the receiver has perfect channel state information. Let us rewrite the received signal in (\ref{yd}) as
\begin{align}
\nonumber
\mathbf{y}_\mathrm{d}=\sqrt{q_\mathrm{u}}\mathbf{\widehat{h}}x_\mathrm{u}+\sqrt{q_\mathrm{j}}\mathbf{\widehat{g}}x_\mathrm{j}+
\sqrt{q_\mathrm{u}}\mathbf{e}_\mathrm{u}x_\mathrm{u}+\sqrt{q_\mathrm{j}}\mathbf{e}_\mathrm{j}x_\mathrm{j}+ \mathbf{n}_\mathrm{d},
\end{align}
where $\mathbf{e}_\mathrm{u}\triangleq \mathbf{h}-\mathbf{\widehat{h}}$ and $\mathbf{e}_\mathrm{j}\triangleq \mathbf{g}-\mathbf{\widehat{g}}$ are the desired and jamming channel estimation errors, respectively. By treating $\mathbf{w}\triangleq\sqrt{q_\mathrm{u}}\mathbf{e}_\mathrm{u}x_\mathrm{u}+\sqrt{q_\mathrm{j}}\mathbf{e}_\mathrm{j}x_\mathrm{j}+ \mathbf{n}_\mathrm{d}$ as equivalent uncorrelated additive Gaussian noise, an MMSE-type receive filter can be obtained as
\begin{align}
\label{ammse}
\mathbf{a}_{\mathrm{MMSE}}=\left(\mathbf{\widehat{g}}\mathbf{\widehat{g}}^H+\frac{1}{q_\mathrm{j}} \boldsymbol{\Psi}\right)^{-1}\mathbf{\widehat{h}},
\end{align}
where $\boldsymbol{\Psi}$ is the covariance matrix of the signal associated with estimation errors plus noise, i.e.,
\begin{align}
\label{psi}
\boldsymbol{\Psi}=\mathbb{E}\{\mathbf{w}\mathbf{w}^H\}=\sigma_\mathrm{e}^2\mathbf{I}_M.
\end{align}
Since $\mathbf{s}_\mathrm{j}$ is unknown, the expectation in (\ref{psi}) is over all random variables including $\mathbf{s}_\mathrm{j}$. Therefore, the equivalent noise variance  $\sigma_\mathrm{e}^2$ is given by
\begin{align}
\nonumber
\sigma_\mathrm{e}^2=q_\mathrm{u}\beta_\mathrm{u}(1-\eta_\mathrm{u}\sqrt{\tau p_\mathrm{u}})+q_\mathrm{j}(\beta_\mathrm{j}(1+\eta_\mathrm{j}^2p_\mathrm{j})+\eta_\mathrm{j}^2\sigma^2)+\sigma^2.
\end{align}
The MMSE-type receive filter in (\ref{ammse}) corresponds to an RZF receive filter with $\mu=\sigma_\mathrm{e}^2/q_\mathrm{j}$. Note that our setup includes the jamming pilot contamination, which makes the equivalent noise $\mathbf{w}$ correlated with the estimated channels $\mathbf{\widehat{h}}$ and $\mathbf{\widehat{g}}$. The receive filter in (\ref{ammse}) is thus not the true MMSE filter, i.e., $\mathbf{a}_{\mathrm{MMSE}}$ may not be optimal in the conventional sense. That is why we call $\mathbf{a}_{\mathrm{MMSE}}$  an ``MMSE-type'' receive filter.

\subsubsection{ZF-type receiver}\label{subsec:zf}
Motivated by the fact that the jamming signal is a main source of limitation in massive MIMO, we consider a receiver that focuses on nulling the jamming signal, i.e., a ZF-type receiver. We show that $\mathbf{a}_\mathrm{RZF}$ corresponds to a ZF-type receiver when $\mu \to 0$.

Following the matrix inversion lemma \cite[Lemma 2]{HBD13JSAC}, the RZF receiver in \eqref{arzf} can be expressed as
\begin{align}
\label{rzf2}
\mathbf{a}_\mathrm{RZF}=\frac{1}{\mu}\left(\mathbf{I}_M-\frac{\mathbf{\widehat{g}}\mathbf{\widehat{g}}^H}{\mu+\|\mathbf{\widehat{g}}\|^2}\right)
\mathbf{\widehat{h}}.
\end{align}
Since a deterministic scalar factor does not change the performance of a linear receive filter (it appears in all terms in the SINR), the RZF receiver in (\ref{rzf2}) is equivalent to
\begin{align}
\label{rzf3}
\mathbf{\tilde{a}}_\mathrm{RZF}=\left(\mathbf{I}_M-\frac{\mathbf{\widehat{g}}\mathbf{\widehat{g}}^H}
{\mu+\|\mathbf{\widehat{g}}\|^2}\right)\mathbf{\widehat{h}}.
\end{align}
Therefore, when $\mu \to 0$ the RZF receiver $\mathbf{a}_\mathrm{RZF}$ becomes a ZF-type receiver $\mathbf{a}_{\mathrm{ZF}}$, which can be expressed as
\begin{align}
\label{azf}
\mathbf{a}_{\mathrm{ZF}}=\left(\mathbf{I}_M-\frac{\mathbf{\widehat{g}}\mathbf{\widehat{g}}^H}{\|\mathbf{\widehat{g}}\|^2}\right)
\mathbf{\widehat{h}}.
\end{align}
Note that we have used the projection-matrix expression for ZF, which is equivalent to using pseudo-inverses.
Due to the imperfect channel estimation, the linear receiver in (\ref{azf}) is not an exact ZF receiver since $\mathbf{\widehat{g}}^H\mathbf{a}_{\mathrm{ZF}}=0$ but $\mathbf{g}^H\mathbf{a}_{\mathrm{ZF}}$ is generally non-zero. Thus, we call the linear receiver in (\ref{azf}) a ``ZF-type'' receiver.

%%%%%%%%%%%%%%%%%%%%%%%%%%%%%%%%%%%%%%%%%%%%%%%%%%%%%%%%%%%%%%%%%%%%%%%%%%%%%%%%%%%%%%%%%%%

%%%%%%%%%%%%%%%%%%%%%%%%%%%%%%%%%%%%%%%%%%%%%%%%%%%%%%%%%%%%%%%%%%%%%%%%%%%%%%%%%%%%%%%%%%%
\section{Performance Analysis}\label{sec:opt_rx}
%%%%%%%%%%%%%%%%%%%%%%%%%%%%%%%%%%%%%%%%%%%%%%%%%%%%%%%%%%%%%%%%%%%%%%%%%%%%%%%%%%%%%%%%%%%

In this section, we analyze the performance of the different proposed receivers and derive the optimal regularization factor for the RZF receiver. The performance of the proposed receivers are also analyzed for systems with extremely strong jammers.

\subsection{Large-scale Approximations}\label{subsec:la}

First of all, let us analyze the effective SINR, which is achieved by the RZF receiver. By exploiting the asymptotic properties of massive MIMO, we can obtain a closed-form \emph{large-scale approximation of the effective SINR} $\rho$, i.e., an approximation of the effective SINR  which almost surely (a.s.) converges to its true value when the number of antennas $M$ tends to infinity. To this end, we use the following notation: for two sequences $f_1[M]$ and $f_2[M]$, we write $f_1[M]\asymp f_2[M]$ to denote that $f_1[M]-f_2[M]\xrightarrow[M\to\infty]{a.s.}0$, where ``$\xrightarrow[M\to\infty]{a.s.}0$'' denotes a.s. convergence. We now obtain the following large-scale approximation.
\begin{thrm}\label{Thrm:rho_rzf}
Assume that the RZF receiver $\mathbf{a}=\mathbf{a}_{\mathrm{MMSE}}$ is used for a fixed $\mu$, then a large-scale approximation of the effective SINR in (\ref{rho}) can be obtained as $\rho \asymp \rho_{\mathrm{RZF}} $, where  $\rho_{\mathrm{RZF}} $ is given by
\begin{align}
\label{rho_rzf}
\rho_{\mathrm{RZF}} = \frac{M\tau p_\mathrm{u}q_\mathrm{u}\beta_\mathrm{u}^2}{\tau p_\mathrm{u}q_\mathrm{u}\beta_\mathrm{u}^2+M\tau p_\mathrm{j}q_\mathrm{j}\delta_1\beta_\mathrm{j}^2\left(\frac{\mu/M+\eta_\mathrm{j}^2\sigma^2}{\mu/M+\eta_\mathrm{j}^2\gamma_\mathrm{j}}\right)^2+\sigma^2(\tau p_\mathrm{u}\beta_\mathrm{u}+\sigma^2
+\tau p_\mathrm{j}\delta_1\beta_\mathrm{j}\frac{\tau p_\mathrm{j}\delta_2 \eta_\mathrm{j}^4\sigma^2\beta_\mathrm{j}+(\mu/M+\eta_\mathrm{j}^2\sigma^2)^2 }{(\mu/M+\eta_\mathrm{j}^2\gamma_\mathrm{j})^2})},
\end{align}
$\delta_1=|\mathbf{s}_\mathrm{j}^T\mathbf{s}_\mathrm{u}^*|^2$, $\delta_2=|\mathbf{s}_\mathrm{j}^T\mathbf{s}_\mathrm{\overline{u}}^*|^2$, and $\gamma_\mathrm{j}=\tau p_\mathrm{j}\delta_2\beta_\mathrm{j}+\sigma^2$.
\end{thrm}

\begin{proof}

When the RZF receiver is used, the power terms of the effective SINR in (\ref{rho}) can be calculated as follows.

{\Large$\cdot$} \emph{The desired signal term $q_\mathrm{u}|\mathbb{E}\{\mathbf{a}^H\mathbf{h}|\mathbf{s}_\mathrm{j}\}|^2$}

Let us consider
\begin{align}
\nonumber
\frac{\mathbf{a}^H\mathbf{h}}{M}&=\frac{\mathbf{\widehat{h}}^H}{M}\left(\mathbf{\widehat{g}}\mathbf{\widehat{g}}^H+
\mu\mathbf{I}_M\right)^{-1}\mathbf{h}\\
\nonumber
&=\frac{\mathbf{\widehat{h}}^H}{M}\frac{1}{\mu}\left(\mathbf{I}_M-
\frac{\mathbf{\widehat{g}}\mathbf{\widehat{g}}^H}{\mu+\|\mathbf{\widehat{g}}\|^2}\right)\mathbf{h}\\
\label{ah}
&=\frac{1}{\mu}\frac{\mathbf{\widehat{h}}^H\mathbf{h}}{M}-\mu
\frac{\frac{\mathbf{\widehat{h}}^H\mathbf{\widehat{g}}}{M}\frac{\mathbf{\widehat{g}}^H\mathbf{h}}{M}}{\frac{\mu}{M}+\frac{\|\mathbf{\widehat{g}}\|^2}{M}},
\end{align}
where the second equality follows from the matrix inversion lemma. Since $\mathbf{h}$, $\mathbf{g}$, and ($\mathbf{n}_1$, $\mathbf{n}_2$) are independent, due to the law of large number, we have the following large-scale approximations
\begin{align}
\label{whh}
\frac{\mathbf{\widehat{h}}^H\mathbf{h}}{M}&=\frac{(\alpha_1 \mathbf{h}+\alpha_2 \mathbf{g}+\mathbf{n}_1)^H\mathbf{h}}{M}\asymp\alpha_1\beta_{\mathrm{u}},\\
\label{whwg}
\frac{\mathbf{\widehat{h}}^H\mathbf{\widehat{g}}}{M}&=\frac{(\alpha_1 \mathbf{h}+\alpha_2 \mathbf{g}+\mathbf{n}_1)^H(\alpha_3\mathbf{g}+\mathbf{n}_2)}{M}\asymp\alpha_2^*\alpha_3\beta_{\mathrm{j}},\\
\label{wgh}
\frac{\mathbf{\widehat{g}}^H\mathbf{h}}{M}&=\frac{(\alpha_3\mathbf{g}+\mathbf{n}_2)^H\mathbf{h}}{M}\asymp0,\\
\label{wgwg}
\frac{\|\mathbf{\widehat{g}}\|^2}{M}&\asymp \eta_{\mathrm{j}}\gamma_{\mathrm{j}}, \\
\frac{\mu}{M} &\asymp 0.
\end{align}
Thus, it follows that $\frac{\mathbf{a}^H\mathbf{h}}{M}\asymp\frac{1}{\mu}\alpha_1\beta_{\mathrm{u}}$ and since a.s. convergence implies convergence in mean we further have
\begin{align}
\label{t1}
\frac{q_\mathrm{u}|\mathbb{E}\{\mathbf{a}^H\mathbf{h}|\mathbf{s}_\mathrm{j}\}|^2}{M^2}\asymp \frac{q_\mathrm{u}}{\mu^2}\alpha_1^2\beta_{\mathrm{u}}^2.
\end{align}

{\Large$\cdot$} \emph{The signal gain uncertainty term} $q_\mathrm{u}\texttt{var}\{\mathbf{a}^H\mathbf{h}|\mathbf{s}_\mathrm{j}\}$

By using the large-scale approximations in (\ref{ah}), (\ref{whwg}), (\ref{wgh}), and (\ref{wgwg}) we have
\begin{align}
\nonumber
\frac{\mathbf{a}^H\mathbf{h}}{M}\asymp\frac{\alpha_1}{\mu}\frac{\mathbf{h}^H\mathbf{h}}{M}.
\end{align}
Thus,
\begin{align}
\nonumber
\texttt{var}\left\{\frac{\mathbf{a}^H\mathbf{h}}{M}  \bigg| \mathbf{s}_\mathrm{j}\right\}&\asymp\frac{\alpha_1^2}{\mu^2M^2}
\texttt{var}\{\mathbf{h}^H\mathbf{h}|\mathbf{s}_\mathrm{j}\}\\
\nonumber
&=\frac{\alpha_1^2}{\mu^2M^2}M\beta_\mathrm{u}^2=\frac{\alpha_1^2}{\mu^2M}\beta_\mathrm{u}^2.
\end{align}
Therefore,
\begin{align}
\label{t2}
\frac{q_\mathrm{u}\texttt{var}\{\mathbf{a}^H\mathbf{h}|\mathbf{s}_\mathrm{j}\}}{M^2}\asymp \frac{q_\mathrm{u}}{M\mu^2}\alpha_1^2\beta_\mathrm{u}^2.
\end{align}

{\Large$\cdot$} \emph{The jamming term $q_\mathrm{j}\mathbb{E}\{|\mathbf{a}^H\mathbf{g}|^2|\mathbf{s}_\mathrm{j}\}$}

By following similar steps as for the desired signal term, we have
\begin{align}
\nonumber
\frac{\mathbf{a}^H\mathbf{g}}{M}\asymp\frac{1}{\mu}\alpha_2^*\beta_{\mathrm{j}}\frac{\mu/M+\eta_{\mathrm{j}}^2\sigma^2}
{\mu/M+\eta_{\mathrm{j}}^2\gamma_\mathrm{j}}
\end{align}
and
\begin{align}
\label{t3}
\frac{q_\mathrm{j}\mathbb{E}\{|\mathbf{a}^H\mathbf{g}|^2|\mathbf{s}_\mathrm{j}\}}{M^2}\asymp \frac{q_\mathrm{j}}{\mu^2}\left(\frac{\mu/M+\eta_{\mathrm{j}}^2\sigma^2}{\mu/M+\eta_{\mathrm{j}}^2\gamma_\mathrm{j}}\right)^2|\alpha_2|^2\beta_\mathrm{j}^2.
\end{align}

{\Large$\cdot$} \emph{The noise term $\sigma^2\mathbb{E}\{\|\mathbf{a}\|^2|\mathbf{s}_\mathrm{j}\}$}

Once again, by following similar steps as for the desired signal term, when $M\to\infty$ we have
\begin{align}
\label{t4}
\frac{\sigma^2\mathbb{E}\{\|\mathbf{a}\|^2|\mathbf{s}_\mathrm{j}\}}{M^2}\asymp
\frac{\sigma^2}{M\mu^2}\left(\alpha_1^2\beta_\mathrm{u}+\eta_\mathrm{u}^2\sigma^2+|\alpha_2|^2\beta_\mathrm{j}\frac{|\alpha_3|^2\eta_\mathrm{j}^2\sigma^2\beta_\mathrm{j}
+(\mu/M+\eta_{\mathrm{j}}^2\sigma^2)^2 }{(\mu/M+\eta_{\mathrm{j}}^2\gamma_\mathrm{j})^2}\right).
\end{align}

Substituting \eqref{t1}, \eqref{t2}, \eqref{t3}, and \eqref{t4} into
\eqref{rho}, we obtain a large-scale approximation of $\rho$ as
\begin{align}
\nonumber
&\rho  \asymp \rho_{\mathrm{RZF}}=\\
\label{rho_rzf1}
&\frac{Mq_\mathrm{u}\alpha_1^2\beta_\mathrm{u}^2}{q_\mathrm{u}\alpha_1^2\beta_\mathrm{u}^2+Mq_\mathrm{j}\left(\frac{\mu/M+\eta_\mathrm{j}^2\sigma^2}{\mu/M+\eta_\mathrm{j}^2\gamma_\mathrm{j}}\right)^2|\alpha_2|^2\beta_\mathrm{j}^2+\sigma^2\left(\alpha_1^2\beta_\mathrm{u}
+\eta_\mathrm{u}^2\sigma^2+|\alpha_2|^2\beta_\mathrm{j}\frac{|\alpha_3|^2\eta_\mathrm{j}^2\sigma^2\beta_\mathrm{j}+(\mu/M+\eta_{\mathrm{j}}^2\sigma^2)^2 }{(\mu/M+\eta_{\mathrm{j}}^2\gamma_\mathrm{j})^2}\right)}.
\end{align}
By denoting $\delta_1=|\mathbf{s}_\mathrm{j}^T\mathbf{s}_\mathrm{u}^*|^2$, $\delta_2=|\mathbf{s}_\mathrm{j}^T\mathbf{s}_\mathrm{\overline{u}}^*|^2$ and using the fact that $\alpha_1^2=\eta_\mathrm{u}^2\tau p_\mathrm{u}$, $|\alpha_2|^2=\eta_\mathrm{u}^2\tau p_\mathrm{j}\delta_1$,  $|\alpha_3|^2=\eta_\mathrm{j}^2\tau p_\mathrm{j}\delta_2$, the effective SINR in (\ref{rho_rzf1}) can be rewritten as in \eqref{rho_rzf}.
This completes the proof for Theorem~\ref{Thrm:rho_rzf}.
\end{proof}

From Theorem~\ref{Thrm:rho_rzf} and its proof, we revisit our discussions in the end of Section~\ref{sec:probsetup} regarding to the impact of different power terms in the effective SINR. We can see that when $M$ is large, the power terms associated with the channel uncertainty and additive noise scale with $M$ and are negligible compared to the desired signal term and the jamming term, which scale with $M^2$.

Moreover, as we have discussed in Remark 1, the effective SINR $\rho_{\mathrm{RZF}}$ is dependent on $\delta_1=|\mathbf{s}_\mathrm{j}^T\mathbf{s}_\mathrm{u}^*|^2$ and $\delta_2=|\mathbf{s}_\mathrm{j}^T\mathbf{s}_\mathrm{\overline{u}}^*|^2$. Although $\mathbf{s}_\mathrm{j}$ is unknown a priori, the BS can estimate $\delta_1$ and $\delta_2$ very accurately thanks to the asymptotic properties of massive MIMO. Following from (\ref{yth}), (\ref{ytg}), and the law of large number, we have
\begin{align}
\label{yapp}
\left\{\begin{array}{l}
  \frac{1}{M}\|\mathbf{y}_\mathrm{u}\|^2 ~\asymp ~\tau p_\mathrm{u}\beta_\mathrm{u}+\tau p_\mathrm{j}\beta_\mathrm{j}\delta_1 + \sigma^2, \\
  \frac{1}{M}\|\mathbf{y}_\mathrm{j}\|^2 ~\asymp ~\tau p_\mathrm{j}\beta_\mathrm{j}\delta_2 + \sigma^2.
\end{array}
\right.
\end{align}
Therefore, the BS can estimate $\delta_1$ and $\delta_2$ as
\begin{align}
\label{delapp}
\left\{\begin{array}{l}
  \widehat{\delta}_1 = \frac{1}{\tau p_\mathrm{j}\beta_\mathrm{j}}\left(\frac{\|\mathbf{y}_\mathrm{u}\|^2}{M}-\tau p_\mathrm{u}\beta_\mathrm{u}-\sigma^2\right), \\
\widehat{\delta}_2 = \frac{1}{\tau p_\mathrm{j}\beta_\mathrm{j}}\left(\frac{\|\mathbf{y}_\mathrm{j}\|^2}{M}-\sigma^2\right).
\end{array}
\right.
\end{align}
When $M$ is large, which is the case in massive MIMO, the estimates in (\ref{delapp}) are very close to the true values of $\delta_1$ and $\delta_2$. Thus, the BS can evaluate the effective SINR $\rho_{\mathrm{RZF}}$ with high accuracy, as we will show in the numerical analysis (in Section~\ref{sec:numerical}).

A closed-form large-scale approximation of the effective SINR achieved by the MMSE-type receiver can be obtained from Theorem~\ref{Thrm:rho_rzf} by setting $\mu=\frac{\sigma_\mathrm{e}^2}{q_\mathrm{j}}$ in (\ref{rho_rzf}), which gives $\rho \asymp \rho_{\mathrm{MMSE}}$, where $\rho_{\mathrm{MMSE}}$ is given by
\begin{align}
\nonumber
&\rho_{\mathrm{MMSE}} = \\
\label{rho_mmse}
&\frac{M\tau p_\mathrm{u}q_\mathrm{u}\beta_\mathrm{u}^2}{\tau p_\mathrm{u}q_\mathrm{u}\beta_\mathrm{u}^2+M\tau p_\mathrm{j}q_\mathrm{j}\delta_1\beta_\mathrm{j}^2\left(\frac{\sigma_\mathrm{e}^2/(Mq_\mathrm{j})+\eta_\mathrm{j}^2\sigma^2}{\sigma_\mathrm{e}^2/(Mq_\mathrm{j})+\eta_\mathrm{j}^2\gamma_\mathrm{j}}\right)^2+\sigma^2\left(\tau p_\mathrm{u}\beta_\mathrm{u}+\sigma^2
+\tau p_\mathrm{j}\delta_1\beta_\mathrm{j}\frac{\tau p_\mathrm{j}\delta_2 \eta_\mathrm{j}^4\sigma^2\beta_\mathrm{j}+(\sigma_\mathrm{e}^2/(Mq_\mathrm{j})+\eta_\mathrm{j}^2\sigma^2)^2 }{(\sigma_\mathrm{e}^2/(Mq_\mathrm{j})+\eta_\mathrm{j}^2\gamma_\mathrm{j})^2}\right)}.
\end{align}
Similarly, a closed-form large-scale approximation of the effective SINR achieved by the ZF-type receiver can be obtained as $\rho \asymp \rho_\mathrm{ZF}$, where $\rho_\mathrm{ZF}$ is given by
\begin{align}
\label{rho_zf}
\rho_\mathrm{ZF} = \frac{M\tau p_\mathrm{u}q_\mathrm{u}\beta_\mathrm{u}^2}{\tau p_\mathrm{u}q_\mathrm{u}\beta_\mathrm{u}^2+
M\tau p_\mathrm{j}q_\mathrm{j}\delta_1\beta_\mathrm{j}^2\frac{\sigma^4}{\gamma_\mathrm{j}^2}+\sigma^2(\tau p_\mathrm{u}\beta_\mathrm{u} +\sigma^2+\tau p_\mathrm{j}\delta_1\beta_\mathrm{j}\frac{\sigma^2}{\gamma_\mathrm{j}})}.
\end{align}
%In the numerical analysis, we will focus on the achievable rates achieved by the MMSE-type and ZF-type receivers, which are characterized by the effective SINRs in (\ref{rho_mmse}) and (\ref{rho_zf}).

Furthermore, when the number of antennas $M$ tends to infinity, both the effective SINRs with MMSE-type and ZF-type receive filters converge to the same finite limit. Let us call this the asymptotic effective SINR $\rho_{\mathrm{asy}}$, then we have
\begin{align}
\label{rho_asy}
\lim_{M\to\infty}\rho_{\mathrm{MMSE}}=\lim_{M\to\infty}\rho_{\mathrm{ZF}}=\lim_{M\to\infty}\rho_{\mathrm{RZF}}\triangleq \rho_{\mathrm{asy}}=\frac{\gamma_\mathrm{j}^2}{\delta_1\sigma^4}
\frac{p_\mathrm{u}q_\mathrm{u}\beta_\mathrm{u}^2}{p_\mathrm{j}q_\mathrm{j}\beta_\mathrm{j}^2}.
\end{align}
This implies that when the number of antennas $M$ grows large, the effective SINR with the RZF receiver for any regularization factor $\mu$ converges to a finite limit, which is independent of $\mu$. Note that, besides the expected ``signal-to-jamming ratio'' term $\frac{p_\mathrm{u}q_\mathrm{u}\beta_\mathrm{u}^2}{p_\mathrm{j}q_\mathrm{j}\beta_\mathrm{j}^2}$ \cite{PRB16WCL}, the asymptotic effective SINR contains the scaling factor
\begin{align}
\frac{\gamma_\mathrm{j}^2}{\delta_1\sigma^4}=\frac{(\tau p_\mathrm{j}\delta_2\beta_\mathrm{j}+\sigma^2)^2}{\delta_1\sigma^4},
\end{align}
which is resulted from the proposed jamming channel estimation scheme.
It is interesting to see that this scaling factor increases with the jamming pilot power $p_\mathrm{j}$, i.e., the benefit of the proposed scheme is greater with stronger jamming pilot signal. Moreover, after some simple mathematical manipulations it can be shown that the asymptotic effective SINR $\rho_{\mathrm{asy}}$ is an increasing function w.r.t. $p_\mathrm{j}$ for $p_\mathrm{j}>\frac{\sigma^4}{\tau^2\beta_\mathrm{j}^2\delta_2^2}$. In other words, the achievable rate increases with the jamming pilot power $p_\mathrm{j}$ when $p_\mathrm{j}$ is higher than a certain level. Intuitively, when $p_\mathrm{j}$ is large enough, the improvement in the estimation quality of jamming channel, resulted from the increase of $p_\mathrm{j}$,  overcomes the degradation in estimation quality of the legitimate channel.
Therefore, it is expected that the proposed receivers can work well or even better with stronger jamming pilot signal.

\begin{remark}
Although we focus on a single user setup, the proposed scheme can be generalized to having multiple users. Indeed, in a system with multiple users, the effective SINR of user $k$ will be similar to the effective SINR in (\ref{rho}) with an additional inter-user interference term in the denominator. This additional term can be written as $\sum_{i\neq k}q_i\mathbb{E}\{|\mathbf{a}_k^H\mathbf{h}_i|^2|\mathbf{s}_\mathrm{j}\}$, where $q_i$, $\mathbf{a}_i$, and $\mathbf{h}_i$ are the data transmit power, receiver filter, and channel gain of user $i$. Assume that the pilot sequences assigned to the users are orthogonal and $\mathbf{a}_k$ is designed as in (\ref{arzf}) (by setting $\mathbf{h}=\mathbf{h}_k$), then $\mathbf{a}_k$ is independent of $\mathbf{h}_i$ for all $i\neq k$. Therefore, $\frac{1}{M^2}\sum_{i\neq k}q_i\mathbb{E}\{|\mathbf{a}_k^H\mathbf{h}_i|^2|\mathbf{s}_\mathrm{j}\}\asymp 0$. In other words, the proposed receiver is robust against inter-user interference when $M\to\infty$ and can also be applied for a multiple users setup.
\end{remark}

\subsection{Optimal RZF Receiver}\label{subsec:opt_rzf}

For a finite number of antennas $M$, the performance of the proposed RZF receiver highly depends on the regularization factor $\mu$. In the following, we derive the optimal $\mu$, which maximize the effective SINR $\rho_\mathrm{RZF}$, i.e., we investigate the optimal RZF receiver. Using the results from Theorem~\ref{Thrm:rho_rzf}, we can prove the following corollary.

\begin{corol}\label{cor:zf_opt}
Assume that the RZF receiver $\mathbf{a}=\mathbf{a}_{\mathrm{RZF}}$ is used. The effective SINR $\rho_\mathrm{RZF}$ in (\ref{rho_rzf}) approaches its maximum when $\mu\to 0$, i.e., the ZF-type receiver is the optimal RZF receiver.
\end{corol}

\begin{proof}
The proof is given in Appendix~\ref{prof:zf_opt}.
\end{proof}

Corollary~\ref{cor:zf_opt} shows that in massive MIMO systems with jamming attacks, the simple ZF-type receiver outperforms other RZF receivers with non-zero $\mu$, including the MMSE-type receiver. This is a surprising result, but understandable since the MMSE-type receiver is not an optimal receiver,  as discussed in Section~\ref{subsec:mmse}. As shown in \cite{PHS05TCOM}, the regularization factor $\mu$ determines the amount of interference (or jamming in our case) remaining at the receiver as compared to the additive noise. Moreover, jamming pilot contamination, resulting in coherent combining of the jamming signal is a main source of limitation in massive MIMO. Therefore, a favorable approach to deal with jamming attacks in massive MIMO is to perform ZF processing and focus on suppressing the jamming signal.

\subsection{Extremely Strong Jamming}\label{subsec:str_jammer}
Next, we investigate the performance of the proposed receivers for a system with extremely strong jamming, where the jamming powers grow without bound. Based on the results from Theorem~\ref{Thrm:rho_rzf}, we can prove the following.
\begin{corol}\label{cor:str_jammer}
Assume that the RZF receiver $\mathbf{a}=\mathbf{a}_{\mathrm{RZF}}$ is used. For $p_\mathrm{j}=\lambda q_\mathrm{j} \to\infty$, where $\lambda$ is a finite strictly positive constant, we have:
\begin{itemize}
\item For $\mu>0$: $\rho_\mathrm{RZF} \to 0$.
\item For $\mu=0$:
$\rho_\mathrm{RZF}=\rho_\mathrm{ZF} \to \frac{M\tau p_\mathrm{u}q_\mathrm{u}\beta_\mathrm{u}^2}{\tau p_\mathrm{u}q_\mathrm{u}\beta_\mathrm{u}^2+\frac{M\delta_1\sigma^4}{\lambda\tau\delta_2^2}+\sigma^2(\tau p_\mathrm{u}\beta_\mathrm{u} +\sigma^2+\frac{\delta_1}{\delta_2}\sigma^2)}>0.$
\end{itemize}

\end{corol}

\begin{proof}
The proof is given in Appendix~\ref{prof:str_jammer}.
\end{proof}

In Corollary~\ref{cor:str_jammer}, we only consider cases when the power of the jamming pilot signal $p_\mathrm{j}$ grows jointly with the power of the jamming data signal $q_\mathrm{j}$. The results for the other trivial cases, i.e., ($p_\mathrm{j}=\mathrm{const.}, q_\mathrm{j} \to\infty $) and ($p_\mathrm{j}\to\infty, q_\mathrm{j} =\mathrm{const.}$) can be  obtained as follows:
\begin{itemize}
\item If $p_\mathrm{j}=\mathrm{const.}, q_\mathrm{j} \to\infty $: $\rho_\mathrm{RZF}$  converges to zero for all values of $\mu$.
\item If $p_\mathrm{j}\to \infty, q_\mathrm{j}=\mathrm{const.}$: $\rho_\mathrm{RZF}$ converges to zero for $\mu>0$ and $\rho_\mathrm{RZF}=\rho_\mathrm{ZF}$ converges to a non-zero finite value for $\mu=0$.
\end{itemize}

Corollary~\ref{cor:str_jammer} implies that the proposed ZF-type receiver works well even under extremely strong jamming, as long as the power $p_\mathrm{j}$ of the jamming pilot signal increases jointly with the power $q_\mathrm{j}$ of the jamming data signal. This is consistent with our analysis in Section~\ref{subsec:la}, which indicated that the proposed scheme can work well and even better with stronger jamming pilot signal.
This behavior is due to the fact that our proposed receivers are constructed based on the jamming channel estimate $\mathbf{\widehat{g}}$, which is obtained from an empty pilot, and its quality is improved when the jamming pilot power $p_\mathrm{j}$ increases. We can thus reject the jamming better and better as $p_\mathrm{j} \to \infty$, but the jamming still remains since also $q_\mathrm{j} \to \infty$.

%%%%%%%%%%%%%%%%%%%%%%%%%%%%%%%%%%%%%%%%%%%%%%%%%%%%%%%%%%%%%%%%%%%%%%%%%%%%%%%%%%%%%%%%%%%

%%%%%%%%%%%%%%%%%%%%%%%%%%%%%%%%%%%%%%%%%%%%%%%%%%%%%%%%%%%%%%%%%%%%%%%%%%%%%%%%%%%%%%%%%%%
\section{Power Allocation}\label{sec:opt_p}
%%%%%%%%%%%%%%%%%%%%%%%%%%%%%%%%%%%%%%%%%%%%%%%%%%%%%%%%%%%%%%%%%%%%%%%%%%%%%%%%%%%%%%%%%%%

In this section, we will show how the legitimate user should allocate its power between the pilot phase and data transmission phase to maximize the achievable rate, in the presence of jamming. We consider two optimization approaches, which use the asymptotic effective SINR $\rho_\mathrm{asy}$ and the ZF-type effective SINR $\rho_\mathrm{ZF}$ as the objective functions for the cases of infinite and finite number of antennas, respectively.

\subsection{Power Allocation for Infinite $M$}
First, we derive the optimal power allocation for the case with an infinite number of antennas. Assuming that the jamming transmit powers $(p_\mathrm{u}, q_\mathrm{u})$ are fixed, the legitimate user aims to maximize $\rho_{\mathrm{asy}}$ by optimally allocating its transmit powers $(p_\mathrm{u}, q_\mathrm{u})$ during the pilot phase and data transmission phase. We assume that the transmit powers of the legitimate user satisfy
\begin{align}
\label{pc:Pu}
\tau p_\mathrm{u}+(T-\tau)q_\mathrm{u}\leq TP_\mathrm{u},
\end{align}
where $P_\mathrm{u}$ is the average power constraint of the legitimate user.
Given the power constraint in (\ref{pc:Pu}), we consider the asymptotically optimal power allocation problem, expressed as
\begin{align}
\label{op:asy}
\underset{p_\mathrm{u},q_\mathrm{u}}{\mathrm{maximize}}~&~~\rho_{\mathrm{asy}}\\
\nonumber
\mathrm{subject~to}&~~\tau p_\mathrm{u} + (T-\tau)q_\mathrm{u} \leq TP_\mathrm{u},\\
\nonumber
    &~~p_\mathrm{u}\geq 0, q_\mathrm{u}\geq 0.
\end{align}
%The solution to (\ref{op:asy}) is given by the following proposition.
\begin{pro}\label{pro:opt_asy}
The solution to (\ref{op:asy}) is given by
\begin{align}
\label{ops:asy}
\left\{\begin{array}{l}
  p_\mathrm{opt}^\mathrm{asy}= \frac{T}{2\tau}P_\mathrm{u},\\
  q_\mathrm{opt}^\mathrm{asy}= \frac{T}{2(T-\tau)}P_\mathrm{u}.
\end{array}
\right.
\end{align}
\end{pro}

\begin{proof}
Following from (\ref{rho_asy}), the objective function $\rho_\mathrm{asy}$ in (\ref{op:asy}) can be rewritten as
\begin{align}
\label{rho_asy2}
\rho_{\mathrm{asy}}=\frac{(\tau p_\mathrm{j} \delta_2 \beta_\mathrm{j} + \sigma^2)^2\beta_\mathrm{u}^2}{p_\mathrm{j}q_\mathrm{j}\delta_1\sigma^4\beta_\mathrm{j}^2}p_\mathrm{u}q_\mathrm{u}.
\end{align}
We note that $\rho_{\mathrm{asy}}$ is a linear function of the multiplication $p_\mathrm{u}q_\mathrm{u}$. Therefore, the optimization problem (\ref{op:asy}) is equivalent to
\begin{align}
\label{op:asy2}
\underset{p_\mathrm{u},q_\mathrm{u}}{\mathrm{maximize}}~&~~p_\mathrm{u}q_\mathrm{u}\\
\nonumber
\mathrm{subject~to}&~~\tau p_\mathrm{u} + (T-\tau)q_\mathrm{u} \leq TP_\mathrm{u},\\
\nonumber
    &~~p_\mathrm{u}\geq 0, q_\mathrm{u}\geq 0.
\end{align}
Since $p_\mathrm{u}q_\mathrm{u}$ is an increasing function w.r.t. $p_\mathrm{u}$ and $q_\mathrm{u}$, the optimal solution is achieved with equality in the first constraint, i.e., $p_\mathrm{u} = (TP_\mathrm{u}-(T-\tau)q_\mathrm{u})/\tau$. Substituting this equality into the objective function in (\ref{op:asy2}), we achieve an equivalent unconstrained optimization problem with a second-order polynomial objective function and a single variable $q_\mathrm{u}$, which can be easily solved to obtain (\ref{ops:asy}).
\end{proof}

Proposition~\ref{pro:opt_asy} shows that, for the given average power budget $P_\mathrm{u}$ and coherence interval $T$, the asymptotic optimal power allocation ($p_\mathrm{opt}^\mathrm{asy}, p_\mathrm{opt}^\mathrm{asy}$) depends only on the pilot length $\tau$. Moreover, we note that $\tau p_\mathrm{opt}^\mathrm{asy} = (T-\tau)p_\mathrm{opt}^\mathrm{asy}=TP_\mathrm{u}/2$. Intuitively, it is optimal to equally divide the transmit energy for the pilot and data transmission phases when the number of antennas is very large ($M \to \infty$). This result confirms the importance of the channel estimation in massive MIMO.

The optimal power allocation problem in (\ref{op:asy}) is designed for a system with infinite number of antennas. However, as it will be shown in the numerical analysis, the performance loss is relative small when the asymptotic optimal power allocation in (\ref{ops:asy}) is employed for a system with a finite number of antennas. Therefore, the power allocation in (\ref{ops:asy}) can be applied as a simple heuristic power allocation that does not depend on the jammer's powers and signal structure.

\subsection{Power Allocation for Finite $M$}

In the following, we consider the power allocation that maximizes the achievable rate for a system with finite number of antennas. We have shown in Section~\ref{sec:opt_rx} that the ZF-type receiver is the optimal RZF receiver. Thus, we focus on optimizing the power allocation for a system with the ZF-type receiver. Accordingly, the optimal power allocation problem can be expressed as
\begin{align}
\label{op:rho_zf}
\underset{p_\mathrm{u},q_\mathrm{u}}{\mathrm{maximize}}~&~~\rho_\mathrm{ZF}\\
\nonumber
\mathrm{subject~to}&~~\tau p_\mathrm{u} + (T-\tau)q_\mathrm{u} \leq TP_\mathrm{u},\\
\nonumber
    &~~p_\mathrm{u}\geq 0, q_\mathrm{u}\geq 0.
\end{align}

The solution for the optimization problem (\ref{op:rho_zf}) is given by the following proposition.

\begin{pro}\label{pro:op_rho_zf}
The solution to (\ref{op:rho_zf}) is given by
\begin{align}
\label{ops:rho_zf}
\left\{\begin{array}{l}
  p_\mathrm{opt}= \frac{\sqrt{\nu^2+\nu\beta_\mathrm{u}\sigma^2TP_\mathrm{u}}-\nu}{\tau\beta_\mathrm{u}\sigma^2},\\
  q_\mathrm{opt}= \frac{\sqrt{\nu+\beta_\mathrm{u}\sigma^2TP_\mathrm{u}}(\sqrt{\nu+\beta_\mathrm{u}\sigma^2TP_\mathrm{u}}-\sqrt{\nu})}{(T-\tau)\beta_\mathrm{u}\sigma^2},
\end{array}
\right.
\end{align}
where
\begin{align} \label{nu-definition}
\nu=\frac{\tau p_\mathrm{j}\beta_\mathrm{j}\delta_1\sigma^4}{\tau p_\mathrm{j} \delta_2 \beta_\mathrm{j} + \sigma^2}\left(\frac{Mq_\mathrm{j}\beta_\mathrm{j}}{\tau p_\mathrm{j} \delta_2 \beta_\mathrm{j} + \sigma^2}+1 \right)+\sigma^4.
\end{align}
\end{pro}

\begin{proof}

Following from (\ref{rho_zf}), the objective function $\rho_{\mathrm{ZF}}$ in (\ref{op:rho_zf}) can be rewritten as
\begin{align}
\label{rho_zf2}
\rho_\mathrm{ZF} = \frac{M\tau\beta_\mathrm{u}^2 p_\mathrm{u}q_\mathrm{u}}{\tau \beta_\mathrm{u}^2 p_\mathrm{u}q_\mathrm{u}+\tau\beta_\mathrm{u}\sigma^2p_\mathrm{u}+\nu},
\end{align}
where $\nu$ is defined in \eqref{nu-definition}. Therefore, the optimal power allocation problem (\ref{rho_zf}) is equivalent to
\begin{align}
\label{op:rho_zf2}
\underset{p_\mathrm{u},q_\mathrm{u}}{\mathrm{maximize}}~&~~\frac{p_\mathrm{u}q_\mathrm{u}}{\tau \beta_\mathrm{u}^2 p_\mathrm{u}q_\mathrm{u}+\tau\beta_\mathrm{u}\sigma^2p_\mathrm{u}+\nu}\\
\nonumber
\mathrm{subject~to}&~~\tau p_\mathrm{u} + (T-\tau)q_\mathrm{u} \leq TP_\mathrm{u},\\
\nonumber
    &~~p_\mathrm{u}\geq 0, q_\mathrm{u}\geq 0.
\end{align}
By dividing both the numerator and denominator of the objective function in (\ref{op:rho_zf2}) by $p_\mathrm{u}q_\mathrm{u}$, an equivalent optimization problem of (\ref{op:rho_zf2}) is obtained as
\begin{align}
\label{op:rho_zf3}
\underset{p_\mathrm{u},q_\mathrm{u}}{\mathrm{minimize}}~&~~\frac{1}{q_\mathrm{u}}\left(\tau\beta_\mathrm{u}\sigma^2+\frac{\nu}{p_\mathrm{u}}\right)\\
\nonumber
\mathrm{subject~to}&~~\tau p_\mathrm{u} + (T-\tau)q_\mathrm{u} \leq TP_\mathrm{u},\\
\nonumber
    &~~p_\mathrm{u}\geq 0, q_\mathrm{u}\geq 0.
\end{align}

We note that the objective function of the optimization problem in (\ref{op:rho_zf3}) is a decreasing function w.r.t. $p_\mathrm{u}$ and $q_\mathrm{u}$. The optimal solution is achieved when having equality in the first constraint, i.e., $p_\mathrm{u} = (TP_\mathrm{u}-(T-\tau)q_\mathrm{u})/\tau$. Substituting this equality into the objective function in (\ref{op:rho_zf3}), we achieve an equivalent optimization problem with single variable $q_\mathrm{u}$, which is convex. Solving this equivalent optimization problem using the Lagrangian multiplier method and Karush-Kuhn-Tucker (KKT) conditions \cite{cvxb}, we obtain the optimal power allocation in \eqref{ops:rho_zf}.
\end{proof}

It is worth noting that the parameter $\nu$ is dependent on $\delta_1=|\mathbf{s}_\mathrm{j}^T\mathbf{s}_\mathrm{u}^*|^2$ and $\delta_2=|\mathbf{s}_\mathrm{j}^T\mathbf{s}_\mathrm{\overline{u}}^*|^2$. Therefore, in order to utilize the optimal power allocation in (\ref{ops:rho_zf}), the legitimate user needs to know the jamming sequence $\mathbf{s}_\mathrm{j}$ or at least the correlation of the jamming sequence with the pilot sequences, i.e., $|\mathbf{s}_\mathrm{j}^T\mathbf{s}_\mathrm{u}^*|^2$ and $|\mathbf{s}_\mathrm{j}^T\mathbf{s}_\mathrm{\overline{u}}^*|^2$. However, the jamming sequence $\mathbf{s}_\mathrm{j}$ is typically unknown at the legitimate user. Therefore, we consider the achievable rate achieved by the power allocation in (\ref{ops:rho_zf})
as an upper bound that can be achieved with perfect knowledge of  $\mathbf{s}_\mathrm{j}$.

Assuming that $\mathbf{s}_\mathrm{j}$ is unknown, we propose a sub-optimal power allocation solution motivated from the solution in (\ref{ops:rho_zf}). For instance, by replacing $\delta_1$ and $\delta_2$ with their mean values, a sub-optimal power allocation solution for the case of finite number of antennas can be used as
\begin{align}
\label{sops:rho_zf}
\left\{\begin{array}{l}
  p_\mathrm{opt}^\mathrm{sub}= \frac{\sqrt{\tilde{\nu}^2+\tilde{\nu}\beta_\mathrm{u}\sigma^2TP_\mathrm{u}}-\tilde{\nu}}{\tau\beta_\mathrm{u}\sigma^2},\\
  q_\mathrm{opt}^\mathrm{sub}= \frac{\sqrt{\tilde{\nu}+\beta_\mathrm{u}\sigma^2TP_\mathrm{u}}(\sqrt{\tilde{\nu}+\beta_\mathrm{u}\sigma^2TP_\mathrm{u}}-\sqrt{\tilde{\nu}})}{(T-\tau)\beta_\mathrm{u}\sigma^2},
\end{array}
\right.
\end{align}
where $\tilde{\nu}$ is an approximation of the parameter $\nu$, given by
\begin{align}
\label{muap}
\tilde{\nu}=\frac{\tau p_\mathrm{j}\beta_\mathrm{j}\mathbb{E}\{\delta_1\}\sigma^4}{\tau p_\mathrm{j} \mathbb{E}\{\delta_2\} \beta_\mathrm{j} + \sigma^2}\left(\frac{Mq_\mathrm{j}\beta_\mathrm{j}}{\tau p_\mathrm{j} \mathbb{E}\{\delta_2\} \beta_\mathrm{j} + \sigma^2}+1 \right)+\sigma^4=\frac{p_\mathrm{j}\beta_\mathrm{j}\sigma^4}{p_\mathrm{j} \beta_\mathrm{j} + \sigma^2}\left(\frac{Mq_\mathrm{j}\beta_\mathrm{j}}{p_\mathrm{j} \beta_\mathrm{j} + \sigma^2}+1 \right)+\sigma^4.
\end{align}
In \eqref{muap} the second equality follows from the assumption that $\mathbf{s}_\mathrm{j}\in\mathbb{C}^{\tau \times 1}$ is uniformly distributed over the unit sphere, which gives $\mathbb{E}\{\delta_1\}=\mathbb{E}\{|\mathbf{s}_\mathrm{j}^T\mathbf{s}_\mathrm{u}^*|^2\}=1/\tau$ and $\mathbb{E}\{\delta_2\}=\mathbb{E}\{|\mathbf{s}_\mathrm{j}^T\mathbf{s}_\mathrm{\overline{u}}^*|^2\}=1/\tau$ \cite{Jin06TIT}.

\begin{remark}
From the results in (\ref{ops:asy}) and (\ref{sops:rho_zf}), some interesting observations can be made.
\begin{itemize}
    \item[(i)] The sub-optimal power allocation in (\ref{sops:rho_zf}) is consistent with the asymptotic optimal power allocation in (\ref{ops:asy}). Indeed, after some simple mathematical manipulations, one can show that ($p_\mathrm{opt}^\mathrm{sub},q_\mathrm{opt}^\mathrm{sub}$) converge to ($p_\mathrm{opt}^\mathrm{asy},q_\mathrm{opt}^\mathrm{asy}$) when $M\to\infty$.

   \item[(ii)] Unlike the asymptotically optimal power allocation in (\ref{ops:asy}), the sub-optimal power allocation ($p_\mathrm{opt}^\mathrm{sub},q_\mathrm{opt}^\mathrm{sub}$) depends on the number of antennas $M$ and the jamming powers $p_\mathrm{j},q_\mathrm{j}$. Therefore, one can expect a better performance from the sub-optimal power allocation, especially when the number of antennas $M$ is large but finite.

  \item[(iii)] A better approximation of $\nu$ can be achieved by using the instantaneous estimates of $\delta_1$ and $\delta_2$ from (\ref{delapp}) instead of their mean values, as done in $\tilde{\nu}$. However, the estimates from (\ref{delapp}) are only available after the pilot phase. Therefore, we can not use them to facilitate a power allocation, which needs to be pre-determined before the pilot phase.
\end{itemize}
\end{remark}

\section{Numerical Analysis} \label{sec:numerical}

In this section, we numerically evaluate the performance of different linear receivers, including the proposed jamming-resistant receivers, in term of the average achievable rates. The average in (\ref{R}) is taken over 10000 realizations of $\mathbf{s}_\mathrm{j}$. We consider a coherence block of $T=200$ symbols, $\tau=3$, $\beta_\mathrm{u}=\beta_\mathrm{u}=1$, and $\sigma^2=1$. For comparison, we also include the rate achieved by the conventional MRC receiver $\mathbf{a}_{\mathrm{mrc}}=\mathbf{\widehat{h}}$, defined in \eqref{MRC-definition}, which does not use the estimate $\mathbf{\widehat{g}}$ \cite{PRB16WCL}.

\subsection{Performance with Different Linear Receivers}

\begin{figure}[!t]
\centering
\includegraphics[width=10cm]{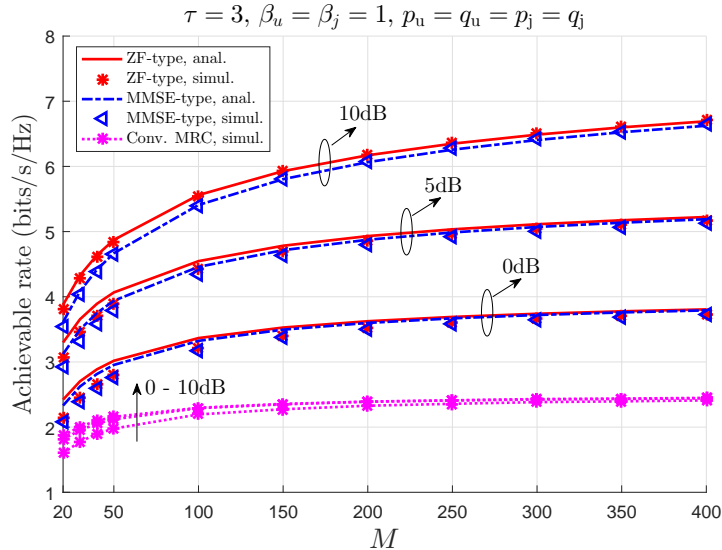}
\caption{Achievable rates for different linear receivers. The solid curves (with label ``anal.'') represent the analysis results. The marked curves (with label ``simul.'') represent the simulation results.}
\label{vsM}
\end{figure}

First, we compare the performance of the different linear receivers. Fig.~\ref{vsM} shows the achievable rates versus the number of antennas at the BS. We assume that $p_\mathrm{u} =q_\mathrm{u}=p_\mathrm{j} =q_\mathrm{j}=\mathrm{SNR}$ and consider different values of $\mathrm{SNR}= 0\,\mathrm{dB},~5\,\mathrm{dB},~10\,\mathrm{dB}$. As expected, the proposed receive filters based on the jamming channel estimate can remarkably improve the system performance, as compared to the MRC receiver. The achievable rates calculated based on the analysis in Theorem~\ref{Thrm:rho_rzf} (curves with ``anal.'') are close to the Monte-Carlo simulations (curves with ``simul.''), and will be asymptotically tight as $M\to\infty$. Moreover, we can see that the simple ZF-type receive filter works particularly well and outperforms the MMSE-type receive filter. This behavior confirms the results from Corollary~\ref{cor:zf_opt} and shows that in massive MIMO, when the jamming effect is critical, a favorable receiver solution is to focus on nulling the jamming signal.

\subsection{Impact of Jamming Powers}
Next, we exemplify the effect of the jamming powers on the system performance. We consider two scenarios with extremely strong jamming. In the first scenario (Fig.~\ref{vspjqj}), we fix the legitimate user's transmit powers as $p_\mathrm{u}=q_\mathrm{u}=5\, \mathrm{ dB}$, and increase both the jammer's transmit powers $p_\mathrm{j}$ and $q_\mathrm{j}$ from $-20\,\mathrm{dB}$ to $40\, \mathrm{dB}$. In the second scenario (Fig.~\ref{vspj}), we fix the legitimate user's transmit powers and the jamming data power  as $p_\mathrm{u}=q_\mathrm{u}=q_\mathrm{j}=5\, \mathrm{ dB}$, and increase only the jamming pilot power $p_\mathrm{j}$. %In other words, we consider the scenario with strong jamming attacks both in the pilot and data transmission phases in
%the first scenario, and the scenario with only a strong jamming attack in the pilot phase in the second scenario.

\begin{figure}[!t]
\centering
\includegraphics[width=10cm]{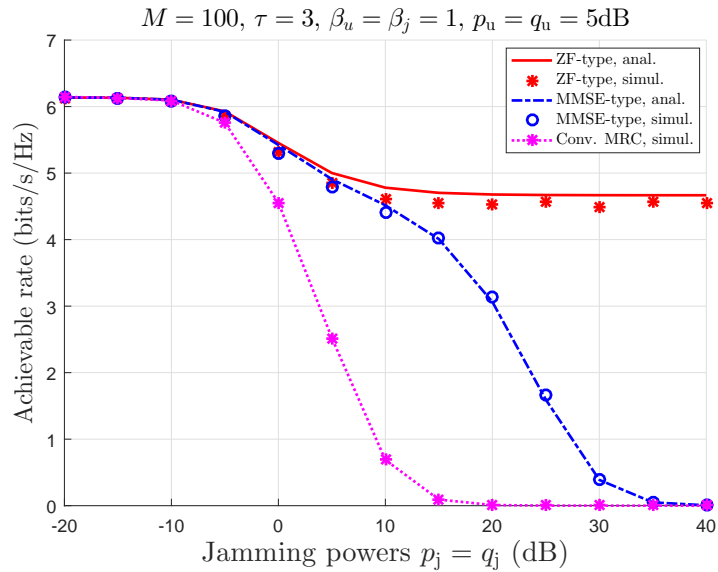}
\caption{Achievable rates for varying jamming powers $p_\mathrm{j}$, $q_\mathrm{j}$.}
\label{vspjqj}
\end{figure}

Fig.~\ref{vspjqj} illustrates the achievable rates in the first scenario, for different values of the jamming attack powers $p_\mathrm{j}=q_\mathrm{j}$. The achievable rates with the conventional MRC and MMSE-type receivers approach zero when $p_\mathrm{j}, q_\mathrm{j} \to \infty$. However, the ZF-type receiver still performs well under strong jamming attacks. Moreover, we can see that the achievable rate with the ZF-type receiver converges to a non-zero value when the jamming data power $q_\mathrm{j}$ tends to infinity, as long as the jamming pilot power $p_\mathrm{j}$ grows proportionally with $q_\mathrm{j}$.

\begin{figure}[!t]
\centering
\includegraphics[width=10cm]{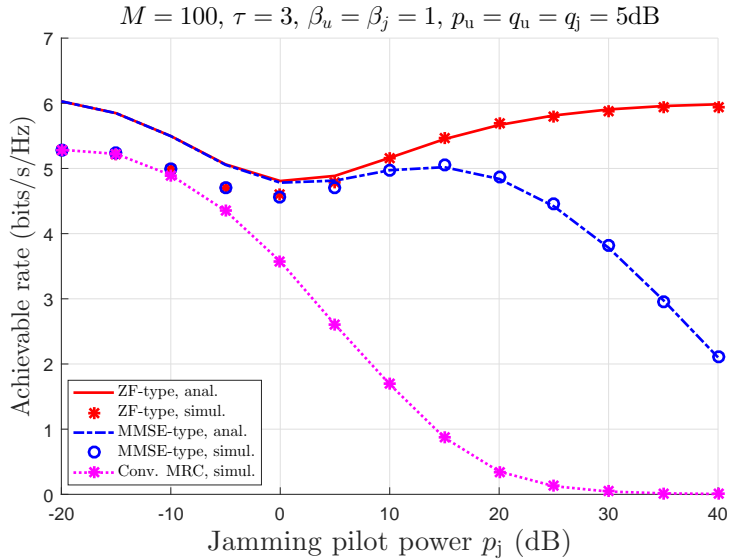}
\caption{Achievable rates for varying jamming pilot power $p_\mathrm{j}$.}
\label{vspj}
\end{figure}

Fig.~\ref{vspj} shows the achievable rates according to the jamming pilot power $p_\mathrm{j}$, in the second scenario with fixed jamming data power. As expected, the achievable rate with the conventional MRC receiver decreases with the increase of the jamming powers. However, the proposed scheme, especially with the ZF-type receiver, still works well with strong jamming attacks. Moreover, it is interesting to see that the ZF-type receiver works better under stronger jamming pilot attacks. We note that there are gaps between our approximate results and the simulation results when $p_\mathrm{j}$ is very small. This is because in our large-scale approximation, we omit some small terms, which a.s.~converge to zero when $M\to \infty$, but still are significant when $p_\mathrm{j}$ is small and $M$ is finite. The gaps will disappear when $M\to \infty$.

The results in Fig.~\ref{vspjqj} and Fig.~\ref{vspj} are consistent with our analysis in Section~\ref{subsec:str_jammer} and can be explained by the fact that the proposed receive filters are constructed using the estimates of both the legitimate channel $\mathbf{h}$ and the jamming channel $\mathbf{g}$. When the jamming pilot power $p_\mathrm{j}$ increases, it does not only degrade the quality of the legitimate channel estimation but also improves the estimation quality of the jamming channel. Thus, the proposed receive filters can improve the system performance if the improvement in the estimation quality of $\mathbf{\widehat{g}}$ overcomes the degradation in the estimation quality of $\mathbf{\widehat{h}}$.

\subsection{Impact of Power Allocation}

\begin{figure}[!t]
\centering
\includegraphics[width=10cm]{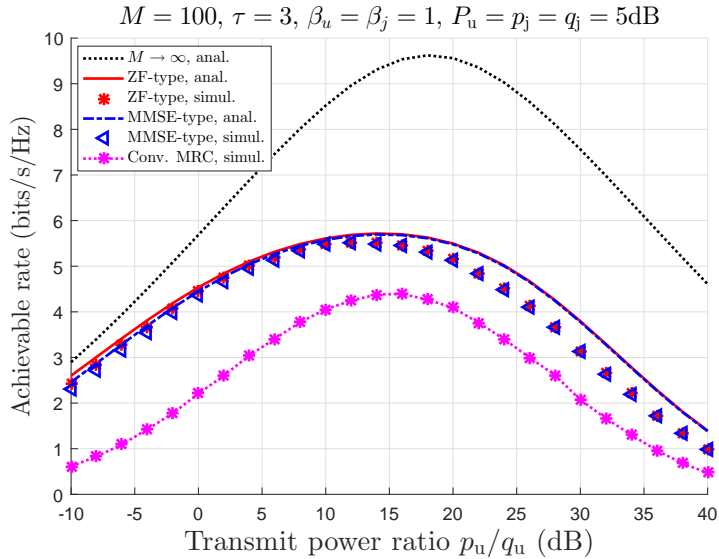}
\caption{Achievable rates for varying power ratios $p_\mathrm{u}/q_\mathrm{u}$. The dotted curve (with label ``$M\to\infty$'') represents the asymptotically achievable rate as $M\to\infty$. The other curves are the achievable rates for $M=100$.}
\label{vs_puqu}
\end{figure}

Next, we evaluate the impact of different power allocations on the system performance.
Fig.~\ref{vs_puqu} illustrates the achievable rates for varying power ratios $p_\mathrm{u}/q_\mathrm{u}$ of the legitimate user. We assume that the average power $P_\mathrm{u}$ and the jamming powers are fixed as $P_\mathrm{u}=p_\mathrm{j}=q_\mathrm{j}=5\, \mathrm{ dB}$. It is observed that the achievable rates have their peak values at certain ratios of $p_\mathrm{u}/q_\mathrm{u}$. In other words, one has to balance between the powers spent by the pilot and data signals to achieve the best transmission rates. Moreover, we can see that the offsets between the asymptotic optimal power ratio (corresponding to the peak value of the asymptotic achievable rate with $M\to\infty$) and the optimal power ratios for finite $M$ (corresponding to the peak values of the achievable rates for MMSE-type and ZF-type receivers with $M=100$) are relatively small. It is thus expected that the asymptotically optimal power allocation can be employed as a simple heuristic power allocation for systems with large but finite number of antennas $M$.

\begin{figure}[!t]
\centering
\includegraphics[width=10cm]{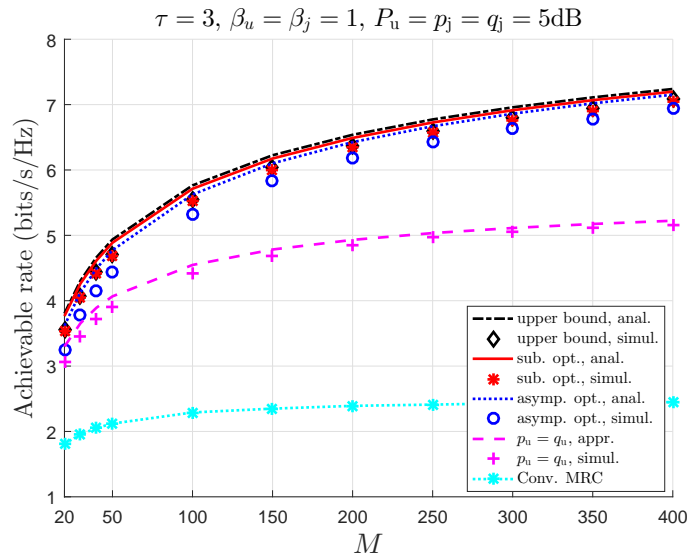}
\caption{Achievable rates with different power allocations.}
\label{vs_optP}
\end{figure}

In Fig.~\ref{vs_optP}, we plot the achievable rates using different power allocations. All the achievable rates are achieved by the ZF-type receiver, except the one with the conventional MRC receiver, which is included for the comparison. In this figure, the ``upper bound'' curves are the achievable rate upper bounds, which are obtained by the optimal power allocation in (\ref{ops:rho_zf}). We can see that the achievable rate increases remarkably when using the simple asymptotically optimal power allocation in (\ref{ops:asy}). Additionally, the achievable rate achieved by the proposed sub-optimal power allocation in (\ref{sops:rho_zf}) is very close to the achievable rate upper bound.

\subsection{Accuracies of the Estimates $\widehat{\delta}_1$ and $\widehat{\delta}_2$}

Lastly, we numerically evaluate the accuracies of the estimations for $\widehat{\delta}_1$ and $\widehat{\delta}_2$ in (\ref{delapp}). We consider $p_\mathrm{j} =q_\mathrm{j} = 5\,\mathrm{dB}$ and $p_\mathrm{u} =q_\mathrm{u} = 0\,\mathrm{dB},~5\,\mathrm{dB},~10\,\mathrm{dB}$. In the upper sub-figure of Fig.~\ref{delta_est}, we plot the  normalized mean square errors (NMSEs) of the estimations in (\ref{delapp}), i.e., $\frac{\mathbb{E}\{|\widehat{\delta}_i-\delta_i|^2\}}{\mathbb{E}\{|\delta_i|^2\}}$, $i=1,2$. In the lower sub-figure of Fig.~\ref{delta_est}, we show the achievable rates, which are calculated based on the exact values of $(\delta_1, \delta_2)$ (the curves with label ``anal.'') and based on the estimates $(\widehat{\delta}_1, \widehat{\delta}_2)$ (the curves with label ``est.'').

We can see that the NMSEs are very small and approach zero when $M$ tends to infinity. The estimation for $\delta_2$ is especially good and independent of $p_\mathrm{u}$ since $\widehat{\delta}_2$ is not impacted by the desired pilot signal. Fig.~\ref{delta_est} also shows that the achievable rates resulted from the estimations in (\ref{delapp}) match very well with the actual achievable rates. Therefore, it is expected that the system can realize the achievable rate in (\ref{R}) without the perfect knowledge of $\mathbf{s}_\mathrm{j}$.

\begin{figure}[!t]
\centering
\includegraphics[width=10.5cm]{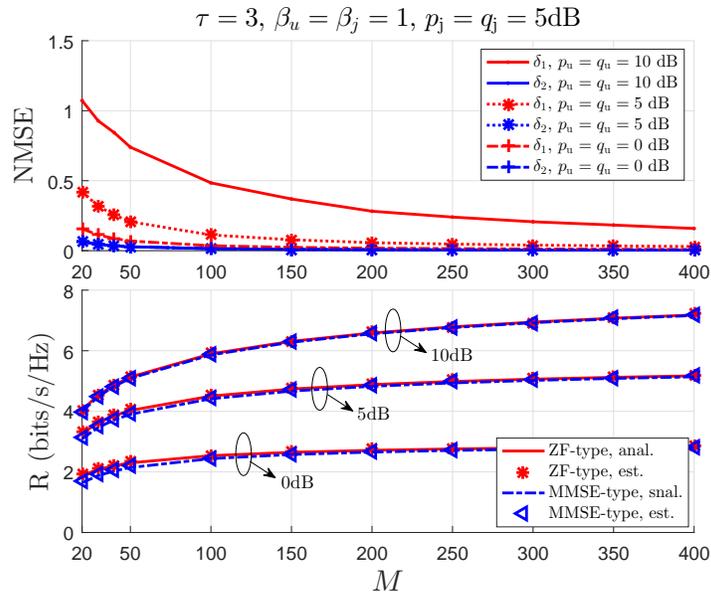}
\caption{Normalized mean square errors of the estimations in (\ref{delapp}) (in upper sub-figure) and the corresponding achievable rates comparison (in the lower sub-figure).}
\label{delta_est}
\end{figure}

\section{Conclusion} \label{sec:conclusion}
A new jamming-resistant receiver approach has been proposed to enhance the robustness of the massive MIMO uplink against jamming attacks. By exploiting purposely unused pilot sequences, the jamming channel can be estimated using the received pilot signal. The results show that the proposed receive filters, which were constructed using the jamming channel estimate, can greatly reduce the effect of jamming attacks and improve the system performance. Moreover, the proposed scheme still works well, or even better, when the jamming powers increase. Due to the critical effect of jamming signal, a ZF-type receive filter, which focuses on nulling the jamming signal is a favorable approach for massive MIMO with jamming. We have also shown that judicious power allocations can substantially improve the performance of the proposed receivers.

\appendices
%%%%%%%%%%%%%%%%%%%%%%%%%%%%%%%%%%%%%%%%%%%%%%%%%%%%%%%%%%%%%%%%%%%%%%%%%%%%%%%%%%%%%%%%%%%

%%%%%%%%%%%%%%%%%%%%%%%%%%%%%%%%%%%%%%%%%%%%%%%%%%%%%%%%%%%%%%%%%%%%%%%%%%%%%%%%%%%%%%%%%%%
\section{Proof of Corollary~\ref{cor:zf_opt}}\label{prof:zf_opt}
%%%%%%%%%%%%%%%%%%%%%%%%%%%%%%%%%%%%%%%%%%%%%%%%%%%%%%%%%%%%%%%%%%%%%%%%%%%%%%%%%%%%%%%%%%%
In order to prove Corollary~\ref{cor:zf_opt}, we need to show that
\begin{align}
\rho_\mathrm{ZF} \geq \rho_\mathrm{RZF}, \forall \mu \geq 0.
\end{align}

Let us consider
\begin{align}
\nonumber
M\tau p_\mathrm{u}q_\mathrm{u}\beta_\mathrm{u}^2\left(\frac{1}{\rho_\mathrm{RZF}}-\frac{1}{\rho_\mathrm{ZF}}\right)=
M\tau p_\mathrm{j}q_\mathrm{j}\delta_1\beta_\mathrm{j}^2\left(\left(\frac{\mu/M+\eta_\mathrm{j}^2\sigma^2}{\mu/M+\eta_\mathrm{j}^2\gamma_\mathrm{j}}\right)^2-\frac{\sigma^4}{\gamma_\mathrm{j}^2}\right)&\\
\label{mmse_zf_1}
+\tau p_\mathrm{j}\delta_1\beta_\mathrm{j}\sigma^2
\left(\frac{\tau p_\mathrm{j}\delta_2 \eta_\mathrm{j}^4\sigma^2\beta_\mathrm{j}+(\mu/M+\eta_\mathrm{j}^2\sigma^2)^2 }{(\mu/M+\eta_\mathrm{j}^2\gamma_\mathrm{j})^2}-\frac{\sigma^2}{\gamma_\mathrm{j}}\right)&.
\end{align}

Using the fact that $\gamma_\mathrm{j}=\tau p_\mathrm{j}\delta_2\beta_\mathrm{j}+\sigma^2$, we can rewrite (\ref{mmse_zf_1}) as
\begin{align}
\nonumber
Mq_\mathrm{u}\alpha_1^2\beta_\mathrm{u}^2\left(\frac{1}{\rho_\mathrm{RZF}}-\frac{1}{\rho_\mathrm{ZF}}\right)=
\tau p_\mathrm{j}q_\mathrm{j}\delta_1\beta_\mathrm{j}^2\mu \frac{\gamma_\mathrm{j}-\sigma^2}{\mu/M+\eta_\mathrm{j}^2\gamma_\mathrm{j}}\left(\frac{\mu/M+\eta_\mathrm{j}^2\sigma^2}{\mu/M+\eta_\mathrm{j}^2\gamma_\mathrm{j}}+\frac{\sigma^2}{\gamma_\mathrm{j}}\right)&\\
\label{mmse_zf_2}
+\tau p_\mathrm{j}\delta_1\beta_\mathrm{j}\sigma^2\frac{\gamma_\mathrm{j}-\sigma^2}{\gamma_\mathrm{j}}\left(\frac{\mu/M}{\mu/M+\eta_\mathrm{j}^2\gamma_\mathrm{j}}\right)^2&.
\end{align}
The right-hand side of (\ref{mmse_zf_2}) is non-negative since $\gamma_\mathrm{j}=\tau p_\mathrm{j}\delta_2\beta_\mathrm{j}+\sigma^2 \geq \sigma^2$. Therefore, we have $\rho_\mathrm{ZF} \geq \rho_\mathrm{RZF}$ for all $\mu\geq 0$.
Thus, the effective SINR $\rho_\mathrm{RZF}$ approaches its maximum when $\mu \to 0$.
%%%%%%%%%%%%%%%%%%%%%%%%%%%%%%%%%%%%%%%%%%%%%%%%%%%%%%%%%%%%%%%%%%%%%%%%%%%%%%%%%%%%%%%%%%%

%%%%%%%%%%%%%%%%%%%%%%%%%%%%%%%%%%%%%%%%%%%%%%%%%%%%%%%%%%%%%%%%%%%%%%%%%%%%%%%%%%%%%%%%%%%
\section{Proof of Corollary~\ref{cor:str_jammer}}\label{prof:str_jammer}
%%%%%%%%%%%%%%%%%%%%%%%%%%%%%%%%%%%%%%%%%%%%%%%%%%%%%%%%%%%%%%%%%%%%%%%%%%%%%%%%%%%%%%%%%%%

The proof for Corollary~\ref{cor:str_jammer} consists of two parts.

\subsection{$\rho_\mathrm{RZF}$ converges to $0$ when $\mu>0$ and $p_\mathrm{j}=\lambda q_\mathrm{j} \to\infty$:}

Let us consider the second term in the denominator of (\ref{rho_rzf}). Since $\gamma_\mathrm{j}=\tau p_\mathrm{j} \delta_2 \beta_\mathrm{j} + \sigma^2$ and
\begin{align}
\nonumber
\eta_{\mathrm{j}}=\frac{1}{{\sqrt{p_\mathrm{j}\beta_\mathrm{j}+\sigma^2}}},
\end{align}
we have
\begin{align}
\nonumber
\left(\frac{\mu/M+\eta_\mathrm{j}^2\sigma^2}{\mu/M+\eta_\mathrm{j}^2\gamma_\mathrm{j}}\right)^2 \to
\frac{\mu^2}{(\mu^2+M\tau\delta_2\beta_\mathrm{j})^2},~~ \mathrm{for} ~~ p_\mathrm{j}=\lambda q_\mathrm{j} \to\infty
\end{align}
and
\begin{align}
\nonumber
M\tau p_\mathrm{j}q_\mathrm{j}\delta_1\beta_\mathrm{j}^2 \left(\frac{\mu/M+\eta_\mathrm{j}^2\sigma^2}{\mu/M+\eta_\mathrm{j}^2\gamma_\mathrm{j}}\right)^2\to
\infty, ~~\mathrm{for} ~~ p_\mathrm{j}=\lambda q_\mathrm{j} \to\infty.
\end{align}

Therefore,
\begin{align}
\nonumber
\rho_\mathrm{RZF} \to 0, ~~ \mathrm{for}~~\mu>0 ~~ \mathrm{and}~~ p_\mathrm{j}=\lambda q_\mathrm{j} \to\infty.
\end{align}

\subsection{$\rho_\mathrm{RZF}$ converges to a non-zero finite value when $\mu=0$ and $p_\mathrm{j}=\lambda q_\mathrm{j} \to\infty$:}

When $\mu=0$, we have $\rho_\mathrm{RZF}=\rho_\mathrm{ZF}$.
Let us consider the second and last terms in the denominator of (\ref{rho_zf}), which are dependent on $p_\mathrm{j}$ and $q_\mathrm{j}$. We have
\begin{align}
\nonumber
M\tau p_\mathrm{j}q_\mathrm{j}\delta_1\beta_\mathrm{j}^2\frac{\sigma^4}{\gamma_\mathrm{j}^2}&=
M\tau p_\mathrm{j}q_\mathrm{j}\delta_1\beta_\mathrm{j}^2\frac{\sigma^4}{(\tau p_\mathrm{j} \delta_2 \beta_\mathrm{j} + \sigma^2)^2}\\
\nonumber
&\to \frac{M\delta_1\sigma^4}{\lambda\tau\delta_2^2},~~\mathrm{for} ~~ p_\mathrm{j}=\lambda q_\mathrm{j} \to\infty
\end{align}
and
\begin{align}
\nonumber
\tau p_\mathrm{j}\delta_1\beta_\mathrm{j}\frac{\sigma^2}{\gamma_\mathrm{j}}&=
\tau p_\mathrm{j}\delta_1\beta_\mathrm{j}\frac{\sigma^2}{\tau p_\mathrm{j} \delta_2 \beta_\mathrm{j} + \sigma^2}\\
\nonumber
&\to \frac{\delta_1}{\delta_2}\sigma^2,~~\mathrm{for} ~~ p_\mathrm{j}=\lambda q_\mathrm{j} \to\infty.
\end{align}

Therefore, we have
\begin{align}
\nonumber
\rho_\mathrm{ZF} \to \frac{M\tau p_\mathrm{u}q_\mathrm{u}\beta_\mathrm{u}^2}{\tau p_\mathrm{u}q_\mathrm{u}\beta_\mathrm{u}^2+\frac{M\delta_1\sigma^4}{\lambda\tau\delta_2^2}+\sigma^2(\tau  p_\mathrm{u}\beta_\mathrm{u} +\sigma^2+\frac{\delta_1}{\delta_2}\sigma^2)}>0, ~~ \mathrm{for} ~~ p_\mathrm{j}=\lambda q_\mathrm{j} \to\infty.
\end{align}

\bibliographystyle{IEEEtran}
\bibliography{IEEEabrv,Treferences}

\end{document}